\documentclass[12pt]{article} % For LaTeX2e
\usepackage[colorlinks, citecolor={blue}]{hyperref}
\usepackage{url}
\usepackage{amsfonts,amscd,amssymb}
\usepackage{amsthm,amsmath}
\usepackage{natbib}
\usepackage{mathtools}
\DeclarePairedDelimiter{\ceil}{\lceil}{\rceil}
\usepackage{algorithm2e}
\usepackage{bm}
\usepackage{bbm} %bb font numbers
\usepackage[table]{xcolor}
\usepackage{verbatim}
\usepackage{graphicx}
\usepackage{setspace}
\usepackage[margin=1in]{geometry}
\usepackage{enumitem}
\usepackage{listings}
\usepackage[textsize=tiny]{todonotes}
\usepackage{tikz}
\usetikzlibrary{shapes.misc}
\usepackage{etoolbox}
\usepackage{appendix}
\usepackage[format=plain,
labelfont={it,bf},
textfont=it]{caption}
\usepackage{subcaption}
\usepackage{wrapfig}
\usepackage{xr}
\usepackage{booktabs}
\usepackage{multirow}
\usepackage{authblk}
\usepackage{mathbbol}
\usepackage{braket}
\usepackage[draft]{minted}

\usetikzlibrary{matrix}
\usetikzlibrary{backgrounds}
\usetikzlibrary{calc}
\usetikzlibrary{arrows,shapes}
\usetikzlibrary{decorations}
\usetikzlibrary{decorations.pathmorphing}
\usetikzlibrary{fit}
\usetikzlibrary{decorations.pathreplacing}

\newtoggle{quickdraw}
\toggletrue{quickdraw} % Uncomment this to render more quickly (non-random)

\definecolor{lightgrey}{rgb}{0.9,0.9,0.9}
\definecolor{darkgreen}{rgb}{0,0.3,0}
\definecolorset{rgb}{}{}{darkred,0.8,0,0;darkgreen,0,0.5,0;darkblue,0,0,0.5}

\SetKwComment{Comment}{/* }{ */}
\newtheorem{prop}{Proposition}

\newcommand{\argmax}{\operatornamewithlimits{arg\,max}}
\newcommand{\argmin}{\operatornamewithlimits{arg\,min}}

\newcommand{\order}[1]{{\cal O}\hspace{-0.2em}\left( #1 \right)}

\definecolor{trevorblue}{rgb}{0.330, 0.484, 0.828}
\definecolor{trevoryellow}{rgb}{0.829, 0.680, 0.306}

\title{A quantum parallel Markov chain Monte Carlo}
\date{}

\author{Andrew J.~Holbrook}

\affil{UCLA Biostatistics}

\graphicspath{{figures/}}

\def\HiLi{\leavevmode\rlap{\hbox to \hsize{\color{trevorblue!15}\leaders\hrule height 0.9\baselineskip depth .5ex\hfill}}}

\makeatletter
\def\PYG@reset{\let\PYG@it=\relax \let\PYG@bf=\relax%
	\let\PYG@ul=\relax \let\PYG@tc=\relax%
	\let\PYG@bc=\relax \let\PYG@ff=\relax}
\def\PYG@tok#1{\csname PYG@tok@#1\endcsname}
\def\PYG@toks#1+{\ifx\relax#1\empty\else%
	\PYG@tok{#1}\expandafter\PYG@toks\fi}
\def\PYG@do#1{\PYG@bc{\PYG@tc{\PYG@ul{%
				\PYG@it{\PYG@bf{\PYG@ff{#1}}}}}}}
\def\PYG#1#2{\PYG@reset\PYG@toks#1+\relax+\PYG@do{#2}}

\@namedef{PYG@tok@w}{\def\PYG@tc##1{\textcolor[rgb]{0.73,0.73,0.73}{##1}}}
\@namedef{PYG@tok@c}{\let\PYG@it=\textit\def\PYG@tc##1{\textcolor[rgb]{0.24,0.48,0.48}{##1}}}
\@namedef{PYG@tok@cp}{\def\PYG@tc##1{\textcolor[rgb]{0.61,0.40,0.00}{##1}}}
\@namedef{PYG@tok@k}{\let\PYG@bf=\textbf\def\PYG@tc##1{\textcolor[rgb]{0.00,0.50,0.00}{##1}}}
\@namedef{PYG@tok@kp}{\def\PYG@tc##1{\textcolor[rgb]{0.00,0.50,0.00}{##1}}}
\@namedef{PYG@tok@kt}{\def\PYG@tc##1{\textcolor[rgb]{0.69,0.00,0.25}{##1}}}
\@namedef{PYG@tok@o}{\def\PYG@tc##1{\textcolor[rgb]{0.40,0.40,0.40}{##1}}}
\@namedef{PYG@tok@ow}{\let\PYG@bf=\textbf\def\PYG@tc##1{\textcolor[rgb]{0.67,0.13,1.00}{##1}}}
\@namedef{PYG@tok@nb}{\def\PYG@tc##1{\textcolor[rgb]{0.00,0.50,0.00}{##1}}}
\@namedef{PYG@tok@nf}{\def\PYG@tc##1{\textcolor[rgb]{0.00,0.00,1.00}{##1}}}
\@namedef{PYG@tok@nc}{\let\PYG@bf=\textbf\def\PYG@tc##1{\textcolor[rgb]{0.00,0.00,1.00}{##1}}}
\@namedef{PYG@tok@nn}{\let\PYG@bf=\textbf\def\PYG@tc##1{\textcolor[rgb]{0.00,0.00,1.00}{##1}}}
\@namedef{PYG@tok@ne}{\let\PYG@bf=\textbf\def\PYG@tc##1{\textcolor[rgb]{0.80,0.25,0.22}{##1}}}
\@namedef{PYG@tok@nv}{\def\PYG@tc##1{\textcolor[rgb]{0.10,0.09,0.49}{##1}}}
\@namedef{PYG@tok@no}{\def\PYG@tc##1{\textcolor[rgb]{0.53,0.00,0.00}{##1}}}
\@namedef{PYG@tok@nl}{\def\PYG@tc##1{\textcolor[rgb]{0.46,0.46,0.00}{##1}}}
\@namedef{PYG@tok@ni}{\let\PYG@bf=\textbf\def\PYG@tc##1{\textcolor[rgb]{0.44,0.44,0.44}{##1}}}
\@namedef{PYG@tok@na}{\def\PYG@tc##1{\textcolor[rgb]{0.41,0.47,0.13}{##1}}}
\@namedef{PYG@tok@nt}{\let\PYG@bf=\textbf\def\PYG@tc##1{\textcolor[rgb]{0.00,0.50,0.00}{##1}}}
\@namedef{PYG@tok@nd}{\def\PYG@tc##1{\textcolor[rgb]{0.67,0.13,1.00}{##1}}}
\@namedef{PYG@tok@s}{\def\PYG@tc##1{\textcolor[rgb]{0.73,0.13,0.13}{##1}}}
\@namedef{PYG@tok@sd}{\let\PYG@it=\textit\def\PYG@tc##1{\textcolor[rgb]{0.73,0.13,0.13}{##1}}}
\@namedef{PYG@tok@si}{\let\PYG@bf=\textbf\def\PYG@tc##1{\textcolor[rgb]{0.64,0.35,0.47}{##1}}}
\@namedef{PYG@tok@se}{\let\PYG@bf=\textbf\def\PYG@tc##1{\textcolor[rgb]{0.67,0.36,0.12}{##1}}}
\@namedef{PYG@tok@sr}{\def\PYG@tc##1{\textcolor[rgb]{0.64,0.35,0.47}{##1}}}
\@namedef{PYG@tok@ss}{\def\PYG@tc##1{\textcolor[rgb]{0.10,0.09,0.49}{##1}}}
\@namedef{PYG@tok@sx}{\def\PYG@tc##1{\textcolor[rgb]{0.00,0.50,0.00}{##1}}}
\@namedef{PYG@tok@m}{\def\PYG@tc##1{\textcolor[rgb]{0.40,0.40,0.40}{##1}}}
\@namedef{PYG@tok@gh}{\let\PYG@bf=\textbf\def\PYG@tc##1{\textcolor[rgb]{0.00,0.00,0.50}{##1}}}
\@namedef{PYG@tok@gu}{\let\PYG@bf=\textbf\def\PYG@tc##1{\textcolor[rgb]{0.50,0.00,0.50}{##1}}}
\@namedef{PYG@tok@gd}{\def\PYG@tc##1{\textcolor[rgb]{0.63,0.00,0.00}{##1}}}
\@namedef{PYG@tok@gi}{\def\PYG@tc##1{\textcolor[rgb]{0.00,0.52,0.00}{##1}}}
\@namedef{PYG@tok@gr}{\def\PYG@tc##1{\textcolor[rgb]{0.89,0.00,0.00}{##1}}}
\@namedef{PYG@tok@ge}{\let\PYG@it=\textit}
\@namedef{PYG@tok@gs}{\let\PYG@bf=\textbf}
\@namedef{PYG@tok@gp}{\let\PYG@bf=\textbf\def\PYG@tc##1{\textcolor[rgb]{0.00,0.00,0.50}{##1}}}
\@namedef{PYG@tok@go}{\def\PYG@tc##1{\textcolor[rgb]{0.44,0.44,0.44}{##1}}}
\@namedef{PYG@tok@gt}{\def\PYG@tc##1{\textcolor[rgb]{0.00,0.27,0.87}{##1}}}
\@namedef{PYG@tok@err}{\def\PYG@bc##1{{\setlength{\fboxsep}{\string -\fboxrule}\fcolorbox[rgb]{1.00,0.00,0.00}{1,1,1}{\strut ##1}}}}
\@namedef{PYG@tok@kc}{\let\PYG@bf=\textbf\def\PYG@tc##1{\textcolor[rgb]{0.00,0.50,0.00}{##1}}}
\@namedef{PYG@tok@kd}{\let\PYG@bf=\textbf\def\PYG@tc##1{\textcolor[rgb]{0.00,0.50,0.00}{##1}}}
\@namedef{PYG@tok@kn}{\let\PYG@bf=\textbf\def\PYG@tc##1{\textcolor[rgb]{0.00,0.50,0.00}{##1}}}
\@namedef{PYG@tok@kr}{\let\PYG@bf=\textbf\def\PYG@tc##1{\textcolor[rgb]{0.00,0.50,0.00}{##1}}}
\@namedef{PYG@tok@bp}{\def\PYG@tc##1{\textcolor[rgb]{0.00,0.50,0.00}{##1}}}
\@namedef{PYG@tok@fm}{\def\PYG@tc##1{\textcolor[rgb]{0.00,0.00,1.00}{##1}}}
\@namedef{PYG@tok@vc}{\def\PYG@tc##1{\textcolor[rgb]{0.10,0.09,0.49}{##1}}}
\@namedef{PYG@tok@vg}{\def\PYG@tc##1{\textcolor[rgb]{0.10,0.09,0.49}{##1}}}
\@namedef{PYG@tok@vi}{\def\PYG@tc##1{\textcolor[rgb]{0.10,0.09,0.49}{##1}}}
\@namedef{PYG@tok@vm}{\def\PYG@tc##1{\textcolor[rgb]{0.10,0.09,0.49}{##1}}}
\@namedef{PYG@tok@sa}{\def\PYG@tc##1{\textcolor[rgb]{0.73,0.13,0.13}{##1}}}
\@namedef{PYG@tok@sb}{\def\PYG@tc##1{\textcolor[rgb]{0.73,0.13,0.13}{##1}}}
\@namedef{PYG@tok@sc}{\def\PYG@tc##1{\textcolor[rgb]{0.73,0.13,0.13}{##1}}}
\@namedef{PYG@tok@dl}{\def\PYG@tc##1{\textcolor[rgb]{0.73,0.13,0.13}{##1}}}
\@namedef{PYG@tok@s2}{\def\PYG@tc##1{\textcolor[rgb]{0.73,0.13,0.13}{##1}}}
\@namedef{PYG@tok@sh}{\def\PYG@tc##1{\textcolor[rgb]{0.73,0.13,0.13}{##1}}}
\@namedef{PYG@tok@s1}{\def\PYG@tc##1{\textcolor[rgb]{0.73,0.13,0.13}{##1}}}
\@namedef{PYG@tok@mb}{\def\PYG@tc##1{\textcolor[rgb]{0.40,0.40,0.40}{##1}}}
\@namedef{PYG@tok@mf}{\def\PYG@tc##1{\textcolor[rgb]{0.40,0.40,0.40}{##1}}}
\@namedef{PYG@tok@mh}{\def\PYG@tc##1{\textcolor[rgb]{0.40,0.40,0.40}{##1}}}
\@namedef{PYG@tok@mi}{\def\PYG@tc##1{\textcolor[rgb]{0.40,0.40,0.40}{##1}}}
\@namedef{PYG@tok@il}{\def\PYG@tc##1{\textcolor[rgb]{0.40,0.40,0.40}{##1}}}
\@namedef{PYG@tok@mo}{\def\PYG@tc##1{\textcolor[rgb]{0.40,0.40,0.40}{##1}}}
\@namedef{PYG@tok@ch}{\let\PYG@it=\textit\def\PYG@tc##1{\textcolor[rgb]{0.24,0.48,0.48}{##1}}}
\@namedef{PYG@tok@cm}{\let\PYG@it=\textit\def\PYG@tc##1{\textcolor[rgb]{0.24,0.48,0.48}{##1}}}
\@namedef{PYG@tok@cpf}{\let\PYG@it=\textit\def\PYG@tc##1{\textcolor[rgb]{0.24,0.48,0.48}{##1}}}
\@namedef{PYG@tok@c1}{\let\PYG@it=\textit\def\PYG@tc##1{\textcolor[rgb]{0.24,0.48,0.48}{##1}}}
\@namedef{PYG@tok@cs}{\let\PYG@it=\textit\def\PYG@tc##1{\textcolor[rgb]{0.24,0.48,0.48}{##1}}}

\def\PYGZob{\char`\{}
\def\PYGZcb{\char`\}}
\def\PYGZca{\char`\^}

\def\PYGZlt{\char`\<}

\def\PYGZhy{\char`\-}

\makeatother

\makeatletter
\def\PYGdefault@reset{\let\PYGdefault@it=\relax \let\PYGdefault@bf=\relax%
	\let\PYGdefault@ul=\relax \let\PYGdefault@tc=\relax%
	\let\PYGdefault@bc=\relax \let\PYGdefault@ff=\relax}
\def\PYGdefault@tok#1{\csname PYGdefault@tok@#1\endcsname}
\def\PYGdefault@toks#1+{\ifx\relax#1\empty\else%
	\PYGdefault@tok{#1}\expandafter\PYGdefault@toks\fi}
\def\PYGdefault@do#1{\PYGdefault@bc{\PYGdefault@tc{\PYGdefault@ul{%
				\PYGdefault@it{\PYGdefault@bf{\PYGdefault@ff{#1}}}}}}}
\def\PYGdefault#1#2{\PYGdefault@reset\PYGdefault@toks#1+\relax+\PYGdefault@do{#2}}

\@namedef{PYGdefault@tok@w}{\def\PYGdefault@tc##1{\textcolor[rgb]{0.73,0.73,0.73}{##1}}}
\@namedef{PYGdefault@tok@c}{\let\PYGdefault@it=\textit\def\PYGdefault@tc##1{\textcolor[rgb]{0.24,0.48,0.48}{##1}}}
\@namedef{PYGdefault@tok@cp}{\def\PYGdefault@tc##1{\textcolor[rgb]{0.61,0.40,0.00}{##1}}}
\@namedef{PYGdefault@tok@k}{\let\PYGdefault@bf=\textbf\def\PYGdefault@tc##1{\textcolor[rgb]{0.00,0.50,0.00}{##1}}}
\@namedef{PYGdefault@tok@kp}{\def\PYGdefault@tc##1{\textcolor[rgb]{0.00,0.50,0.00}{##1}}}
\@namedef{PYGdefault@tok@kt}{\def\PYGdefault@tc##1{\textcolor[rgb]{0.69,0.00,0.25}{##1}}}
\@namedef{PYGdefault@tok@o}{\def\PYGdefault@tc##1{\textcolor[rgb]{0.40,0.40,0.40}{##1}}}
\@namedef{PYGdefault@tok@ow}{\let\PYGdefault@bf=\textbf\def\PYGdefault@tc##1{\textcolor[rgb]{0.67,0.13,1.00}{##1}}}
\@namedef{PYGdefault@tok@nb}{\def\PYGdefault@tc##1{\textcolor[rgb]{0.00,0.50,0.00}{##1}}}
\@namedef{PYGdefault@tok@nf}{\def\PYGdefault@tc##1{\textcolor[rgb]{0.00,0.00,1.00}{##1}}}
\@namedef{PYGdefault@tok@nc}{\let\PYGdefault@bf=\textbf\def\PYGdefault@tc##1{\textcolor[rgb]{0.00,0.00,1.00}{##1}}}
\@namedef{PYGdefault@tok@nn}{\let\PYGdefault@bf=\textbf\def\PYGdefault@tc##1{\textcolor[rgb]{0.00,0.00,1.00}{##1}}}
\@namedef{PYGdefault@tok@ne}{\let\PYGdefault@bf=\textbf\def\PYGdefault@tc##1{\textcolor[rgb]{0.80,0.25,0.22}{##1}}}
\@namedef{PYGdefault@tok@nv}{\def\PYGdefault@tc##1{\textcolor[rgb]{0.10,0.09,0.49}{##1}}}
\@namedef{PYGdefault@tok@no}{\def\PYGdefault@tc##1{\textcolor[rgb]{0.53,0.00,0.00}{##1}}}
\@namedef{PYGdefault@tok@nl}{\def\PYGdefault@tc##1{\textcolor[rgb]{0.46,0.46,0.00}{##1}}}
\@namedef{PYGdefault@tok@ni}{\let\PYGdefault@bf=\textbf\def\PYGdefault@tc##1{\textcolor[rgb]{0.44,0.44,0.44}{##1}}}
\@namedef{PYGdefault@tok@na}{\def\PYGdefault@tc##1{\textcolor[rgb]{0.41,0.47,0.13}{##1}}}
\@namedef{PYGdefault@tok@nt}{\let\PYGdefault@bf=\textbf\def\PYGdefault@tc##1{\textcolor[rgb]{0.00,0.50,0.00}{##1}}}
\@namedef{PYGdefault@tok@nd}{\def\PYGdefault@tc##1{\textcolor[rgb]{0.67,0.13,1.00}{##1}}}
\@namedef{PYGdefault@tok@s}{\def\PYGdefault@tc##1{\textcolor[rgb]{0.73,0.13,0.13}{##1}}}
\@namedef{PYGdefault@tok@sd}{\let\PYGdefault@it=\textit\def\PYGdefault@tc##1{\textcolor[rgb]{0.73,0.13,0.13}{##1}}}
\@namedef{PYGdefault@tok@si}{\let\PYGdefault@bf=\textbf\def\PYGdefault@tc##1{\textcolor[rgb]{0.64,0.35,0.47}{##1}}}
\@namedef{PYGdefault@tok@se}{\let\PYGdefault@bf=\textbf\def\PYGdefault@tc##1{\textcolor[rgb]{0.67,0.36,0.12}{##1}}}
\@namedef{PYGdefault@tok@sr}{\def\PYGdefault@tc##1{\textcolor[rgb]{0.64,0.35,0.47}{##1}}}
\@namedef{PYGdefault@tok@ss}{\def\PYGdefault@tc##1{\textcolor[rgb]{0.10,0.09,0.49}{##1}}}
\@namedef{PYGdefault@tok@sx}{\def\PYGdefault@tc##1{\textcolor[rgb]{0.00,0.50,0.00}{##1}}}
\@namedef{PYGdefault@tok@m}{\def\PYGdefault@tc##1{\textcolor[rgb]{0.40,0.40,0.40}{##1}}}
\@namedef{PYGdefault@tok@gh}{\let\PYGdefault@bf=\textbf\def\PYGdefault@tc##1{\textcolor[rgb]{0.00,0.00,0.50}{##1}}}
\@namedef{PYGdefault@tok@gu}{\let\PYGdefault@bf=\textbf\def\PYGdefault@tc##1{\textcolor[rgb]{0.50,0.00,0.50}{##1}}}
\@namedef{PYGdefault@tok@gd}{\def\PYGdefault@tc##1{\textcolor[rgb]{0.63,0.00,0.00}{##1}}}
\@namedef{PYGdefault@tok@gi}{\def\PYGdefault@tc##1{\textcolor[rgb]{0.00,0.52,0.00}{##1}}}
\@namedef{PYGdefault@tok@gr}{\def\PYGdefault@tc##1{\textcolor[rgb]{0.89,0.00,0.00}{##1}}}
\@namedef{PYGdefault@tok@ge}{\let\PYGdefault@it=\textit}
\@namedef{PYGdefault@tok@gs}{\let\PYGdefault@bf=\textbf}
\@namedef{PYGdefault@tok@gp}{\let\PYGdefault@bf=\textbf\def\PYGdefault@tc##1{\textcolor[rgb]{0.00,0.00,0.50}{##1}}}
\@namedef{PYGdefault@tok@go}{\def\PYGdefault@tc##1{\textcolor[rgb]{0.44,0.44,0.44}{##1}}}
\@namedef{PYGdefault@tok@gt}{\def\PYGdefault@tc##1{\textcolor[rgb]{0.00,0.27,0.87}{##1}}}
\@namedef{PYGdefault@tok@err}{\def\PYGdefault@bc##1{{\setlength{\fboxsep}{\string -\fboxrule}\fcolorbox[rgb]{1.00,0.00,0.00}{1,1,1}{\strut ##1}}}}
\@namedef{PYGdefault@tok@kc}{\let\PYGdefault@bf=\textbf\def\PYGdefault@tc##1{\textcolor[rgb]{0.00,0.50,0.00}{##1}}}
\@namedef{PYGdefault@tok@kd}{\let\PYGdefault@bf=\textbf\def\PYGdefault@tc##1{\textcolor[rgb]{0.00,0.50,0.00}{##1}}}
\@namedef{PYGdefault@tok@kn}{\let\PYGdefault@bf=\textbf\def\PYGdefault@tc##1{\textcolor[rgb]{0.00,0.50,0.00}{##1}}}
\@namedef{PYGdefault@tok@kr}{\let\PYGdefault@bf=\textbf\def\PYGdefault@tc##1{\textcolor[rgb]{0.00,0.50,0.00}{##1}}}
\@namedef{PYGdefault@tok@bp}{\def\PYGdefault@tc##1{\textcolor[rgb]{0.00,0.50,0.00}{##1}}}
\@namedef{PYGdefault@tok@fm}{\def\PYGdefault@tc##1{\textcolor[rgb]{0.00,0.00,1.00}{##1}}}
\@namedef{PYGdefault@tok@vc}{\def\PYGdefault@tc##1{\textcolor[rgb]{0.10,0.09,0.49}{##1}}}
\@namedef{PYGdefault@tok@vg}{\def\PYGdefault@tc##1{\textcolor[rgb]{0.10,0.09,0.49}{##1}}}
\@namedef{PYGdefault@tok@vi}{\def\PYGdefault@tc##1{\textcolor[rgb]{0.10,0.09,0.49}{##1}}}
\@namedef{PYGdefault@tok@vm}{\def\PYGdefault@tc##1{\textcolor[rgb]{0.10,0.09,0.49}{##1}}}
\@namedef{PYGdefault@tok@sa}{\def\PYGdefault@tc##1{\textcolor[rgb]{0.73,0.13,0.13}{##1}}}
\@namedef{PYGdefault@tok@sb}{\def\PYGdefault@tc##1{\textcolor[rgb]{0.73,0.13,0.13}{##1}}}
\@namedef{PYGdefault@tok@sc}{\def\PYGdefault@tc##1{\textcolor[rgb]{0.73,0.13,0.13}{##1}}}
\@namedef{PYGdefault@tok@dl}{\def\PYGdefault@tc##1{\textcolor[rgb]{0.73,0.13,0.13}{##1}}}
\@namedef{PYGdefault@tok@s2}{\def\PYGdefault@tc##1{\textcolor[rgb]{0.73,0.13,0.13}{##1}}}
\@namedef{PYGdefault@tok@sh}{\def\PYGdefault@tc##1{\textcolor[rgb]{0.73,0.13,0.13}{##1}}}
\@namedef{PYGdefault@tok@s1}{\def\PYGdefault@tc##1{\textcolor[rgb]{0.73,0.13,0.13}{##1}}}
\@namedef{PYGdefault@tok@mb}{\def\PYGdefault@tc##1{\textcolor[rgb]{0.40,0.40,0.40}{##1}}}
\@namedef{PYGdefault@tok@mf}{\def\PYGdefault@tc##1{\textcolor[rgb]{0.40,0.40,0.40}{##1}}}
\@namedef{PYGdefault@tok@mh}{\def\PYGdefault@tc##1{\textcolor[rgb]{0.40,0.40,0.40}{##1}}}
\@namedef{PYGdefault@tok@mi}{\def\PYGdefault@tc##1{\textcolor[rgb]{0.40,0.40,0.40}{##1}}}
\@namedef{PYGdefault@tok@il}{\def\PYGdefault@tc##1{\textcolor[rgb]{0.40,0.40,0.40}{##1}}}
\@namedef{PYGdefault@tok@mo}{\def\PYGdefault@tc##1{\textcolor[rgb]{0.40,0.40,0.40}{##1}}}
\@namedef{PYGdefault@tok@ch}{\let\PYGdefault@it=\textit\def\PYGdefault@tc##1{\textcolor[rgb]{0.24,0.48,0.48}{##1}}}
\@namedef{PYGdefault@tok@cm}{\let\PYGdefault@it=\textit\def\PYGdefault@tc##1{\textcolor[rgb]{0.24,0.48,0.48}{##1}}}
\@namedef{PYGdefault@tok@cpf}{\let\PYGdefault@it=\textit\def\PYGdefault@tc##1{\textcolor[rgb]{0.24,0.48,0.48}{##1}}}
\@namedef{PYGdefault@tok@c1}{\let\PYGdefault@it=\textit\def\PYGdefault@tc##1{\textcolor[rgb]{0.24,0.48,0.48}{##1}}}
\@namedef{PYGdefault@tok@cs}{\let\PYGdefault@it=\textit\def\PYGdefault@tc##1{\textcolor[rgb]{0.24,0.48,0.48}{##1}}}

\makeatother

\begin{document}

\maketitle

\begin{abstract}

We propose a novel hybrid quantum computing strategy for parallel MCMC algorithms that generate multiple proposals at each step. This strategy makes the rate-limiting step within parallel MCMC amenable to quantum parallelization by using the Gumbel-max trick to turn the generalized accept-reject step into a discrete optimization problem.  When combined with new insights from the parallel MCMC literature, such an approach allows us to embed target density evaluations within a well-known extension of Grover's quantum search algorithm.  Letting $P$ denote the number of proposals in a single MCMC iteration, the combined strategy reduces the number of target evaluations required from $\order{P}$ to $\order{P^{1/2}}$.  In the following, we review the rudiments of quantum computing, quantum search and the Gumbel-max trick in order to elucidate their combination for as wide a readership as possible.

\end{abstract}

\section{Introduction}\label{sec:intro}

\newcommand{\ttheta}{\boldsymbol{\theta}}
\newcommand{\Ttheta}{\boldsymbol{\Theta}}
\newcommand{\dd}{\mbox{d}}
\newcommand{\ppsi}{\boldsymbol{\psi}}
\newcommand{\U}{\mathbf{U}}
\newcommand{\I}{\mathbf{I}}
\renewcommand{\H}{\mathbf{H}}
\newcommand{\A}{\mathbf{A}}
\newcommand{\B}{\mathbf{B}}
\newcommand{\G}{\mathbf{G}}
\newcommand{\ppi}{\boldsymbol{\pi}}
\newcommand{\llambda}{\boldsymbol{\lambda}}

Parallel MCMC techniques use multiple proposals to obtain efficiency gains over MCMC algorithms such as Metropolis-Hastings \citep{metropolis1953equation,hastings1970monte} and its progeny that use only a single proposal.  \citet{neal2003markov} first develops efficient MCMC transitions for inferring the states of hidden Markov models by proposing a `pool' of candidate states and using dynamic programming to select among them. Next, \citet{tjelmeland2004using} considers inference in the general setting and shows how to maintain detailed balance for an arbitrary number $P$ of proposals.  Consider a probability distribution $\pi(\dd\ttheta)$ defined on $\mathbb{R}^D$ that admits a probability density $\pi(\ttheta)$ with respect to the Lebesgue measure, i.e., $\pi(\dd \ttheta)=:\pi(\ttheta)\dd\ttheta$. To generate samples from the target distribution $\pi$, we craft a kernel $P(\ttheta_0,\dd \ttheta)$ that satisfies
\begin{align}\label{eq:fixes}
\pi(A) = \int \pi(\dd \ttheta_0) P(\ttheta_0,A) 
\end{align} 
for all $A \subset \mathbb{R}^D$.  Letting $\ttheta_0$ denote the current state of the Markov chain, \citet{tjelmeland2004using} proposes sampling from such a kernel $P(\ttheta_0,\cdot)$ by drawing $P$ proposals $\Ttheta_{-0}=(\ttheta_1,\dots,\ttheta_P)$ from a joint distribution $Q(\ttheta_0,\dd \Ttheta_{-0}) =: q(\ttheta_0, \Ttheta_{-0})\dd  \Ttheta_{-0}$ and selecting the next Markov chain state from among the current and proposed states with probabilities
\begin{align}\label{eq:probs}
\pi_p = \frac{\pi(\ttheta_p) q(\ttheta_p, \Ttheta_{-p})}{\sum_{p'=0}^P \pi(\ttheta_{p'})q(\ttheta_{p'}, \Ttheta_{-p'})} \, , \quad p=0,\dots,P \, .
\end{align}
Here, if $\Ttheta=(\ttheta_0,\dots,\ttheta_{P})$ is the $D\times (P+1)$ matrix having the current state and all $P$ proposals as columns, then $\Ttheta_{-p}$ is the $D\times P$ matrix obtained by removing the $p$th column.
\citet{tjelmeland2004using} shows that the kernel $P(\ttheta_0,\cdot)$ constructed in such a manner maintains detailed balance and hence satisfies \eqref{eq:fixes}.  Others have since built on this work, developing parallel MCMC methods that generate or recycle multiple proposals \citep{frenkel2004speed,delmas2009does,calderhead2014general,yang2018parallelizable,luo2019multiple,schwedes2021rao,holbrook2021generating}.  Most recently, \citet{glatt} place all these developments in their natural measure theoretic context, allowing one to trivially apply \eqref{eq:fixes} and \eqref{eq:probs} to distributions over discrete-valued random variables.

Taken together, these developments demonstrate the ability of parallel MCMC methods to alleviate inferential challenges such as multimodality and to deliver performance gains over single-proposal competitors as measured by reduced autocorrelation between samples.  In the following, we focus on the specification presented in Algorithm \ref{alg:pMCMC} but note that the techniques we present may also be effective for generalizations of this basic algorithm.

\begin{figure}[!t]
	\centering
	\scalebox{0.9}{
\begin{minipage}{1\textwidth}
\begin{algorithm}[H]
	\caption{Parallel (multiproposal) MCMC \citep{tjelmeland2004using}}\label{alg:pMCMC}
	\KwData{Initial Markov chain state $\ttheta^{(0)}$; total length of Markov chain $S$; total number of proposals per iteration $P$; routine for evaluating target density $\pi(\cdot)$; routines for drawing random samples from the proposal distribution $Q(\ttheta^{(s)},\cdot)$ and from a $P+1$ discrete distribution $Discrete(\ppi)$ given some probability vector $\ppi=(\pi_0,\dots,\pi_P)$.}
	\KwResult{A Markov chain $\ttheta^{(1)}, \dots, \ttheta^{(S)}$.}
	\For{$s \in\{1,\dots,S\}$}{
		%$\sigma_\pi \gets 0$\;
		$\ttheta_0 \gets \ttheta^{(s-1)}$\;
		$\Ttheta_{-0} \gets Q(\ttheta_0,\cdot)$\;
		\For{$p \in \{0,\dots,P\}$}{
			$\pi_{p} \gets \pi(\ttheta_{p})\, q(\ttheta_{p},\Ttheta_{-p})$\; 
		}
		$\hat{p} \gets Discrete(\ppi/\ppi^T\mathbb{1})$\;
		$\ttheta^{(s)} \gets \ttheta_{\hat{p}}$\;
	}
	\Return{$\ttheta^{(1)}, \dots, \ttheta^{(S)}$}\ .
	\vspace{0.5em}
	
\end{algorithm}
\end{minipage}
}
\end{figure}

One need not use parallel computing to implement Algorithm \ref{alg:pMCMC}, but the real promise and power of parallel MCMC comes from its natural parallelizability \citep{calderhead2014general}. 
  Contemporary hardware design emphasizes architectures that enable execution of multiple mathematical operations simultaneously. Parallel MCMC techniques stand to leverage technological developments that keep modern computation on track with Moore's Law, which predicts that processing power doubles every two years.  For example, the algorithm of \citet{tjelmeland2004using} generates $P$ conditionally independent proposals and then evaluates the probabilities of \eqref{eq:probs}.  One may parallelize the proposal generation step using parallel pseudorandom number generators (PRNG) such as those advanced in \citet{salmon2011parallel}. The computational complexity of the target evaluations $\pi(\ttheta_p)$ is linear in the number of proposals. This presents a significant burden when $\pi(\cdot)$ is computationally expensive, e.g., in big data settings or for Bayesian inversion, but evaluation of the target density over the $P$ proposals is again a naturally parallelizable task.  Moreover, widely available machine learning software such as \textsc{TensorFlow} allows users to easily parallelize both random number generation and target evaluations on general purpose graphics processing units (GPU) \citep{lao2020tfp}. Finally, when generating independent proposals using a proposal distribution of the form $q(\ttheta_0,\Ttheta_{-0})=\prod_{p=1}^Pq(\ttheta_0,\ttheta_{p})$, the acceptance probabilities \eqref{eq:probs} require the $\order{P^2}$ evaluation of the $P+1\choose 2$ terms $q(\ttheta_{p},\ttheta_{p'})$, but \citet{massive,holbrook2021scalable} demonstrate the natural parallelizability of such pairwise operations, obtaining orders-of-magnitude speedups with contemporary GPUs.  The proposed method directly addresses the acceptance step of Algorithm \ref{alg:pMCMC}, while leaving the choice of parallelizing (or not parallelizing) the proposal step to the practitioner.

While parallel MCMC algorithms are increasingly well-suited for developing many-core computational architectures, there are trade-offs that need to be taken into account when choosing how to allocate computational resources.  On one end of the spectrum, \citet{gelman1992inference,gelman1992single} demonstrate the usefulness of generating, combining and comparing multiple independent Markov chains that target the same distribution, and one may perform this embarrassingly parallel task by assigning the operations for each individual chain to a separate central processing unit (CPU) core or GPU work-group.  In this multi-chain context, simultaneously running multiple parallel MCMC chains could limit resources available for the within-chain parallelization described in the previous paragraph.  On the other end of the spectrum, one may find it useful to allocate resources to overcome computational bottlenecks within a single Markov chain that uses a traditional accept-reject step. In big data contexts, \citet{massive,holbrook2021scalable,holbrook2022bayesian,holbrook2022viral} use multi- and many-core processing to accelerate log-likelihood and log-likelihood gradient calculations within single Metropolis-Hastings and Hamiltonian Monte Carlo \citep{neal2011mcmc} generated chains.  This big data strategy might again limit resources available for parallelization across proposals and target evaluations described above.

In the presence of these trade-offs, it is worthwhile to develop additional parallel computing tools for parallel MCMC so that future scientists may be better able to flexibly assign sub-routines to different computing resources in a way that is tailored to their specific inference task.  In particular, we show that quantum parallelization can be a useful tool for parallel MCMC when evaluation of the target density represents the computational bottleneck.

Quantum computers use quantum mechanics to store information and manipulate data. By leveraging the idiosyncracies of quantum physics, these computers are able to deliver speedups over conventional computing for a relatively small number of computational problems. Some of these speedups are very large. Quantum computers can achieve \emph{exponential} speedups over conventional computers for some computational tasks..  Shor's quantum algorithm for factoring an integer $N$ \citep{shor} is polynomial in $\log N$ compared to a super-polynomial classical optimum. The HHL algorithm for solving sparse $N$-dimensional linear systems \citep{harrow2009quantum} is $\order{\log N}$ compared to a classical $\order{N}$. Other algorithms deliver a still impressive \emph{polynomial} speedup.  For example, the algorithms considered in the following (Section \ref{sec:qmin}) achieve quadratic speedups over conventional computing, turning  $\order{N}$ runtimes into $\mathcal{O}(\sqrt{ N})$.  
 Both the magnitude of quantum computing's successes and the relative rarity of those successes mean that there is a general interest in leveraging quantum algorithms for previously unconsidered computational problems.  We will show that---with the help of the Gumbel-max trick (Section \ref{sec:gm})---established quantum optimization techniques are actually useful for sampling from computationally expensive discrete distributions and apply this insight to achieving quadratic speedups for parallel MCMC.

We provide a short introduction to the most basic elements of quantum computing in Appendix \ref{sec:shortIntro}. The interested reader may look to \citet{nielsen2002quantum} for a canonical introduction to quantum algorithms or \citet{lopatnikova2021introduction,wang2022quantum} for surveys geared toward statisticians.   In Section \ref{sec:qSearch}, we review and compare three different quantum search algorithms that enjoy quadratic speedups over conventional algorithms.  In Section \ref{sec:qmin} we review the quantum minimization algorithm of \citet{durr1996quantum} and investigate the use of warm-starting therein.  We then combine the Gumbel-max trick (Section \ref{sec:gm}) and recent insights in parallel MCMC with the quantum minimization algorithm to create the quantum parallel MCMC algorithm for both continuously-valued (Section \ref{sec:continu}) and discrete-valued (Section \ref{sec:discrete}) targets.

\section{Quantum search and quantum minimization}\label{sec:searchAndMin}

\citet{grover1996fast} demonstrates how use a quantum computer to find a single marked element within a finite set of $N$ items with only $\mathcal{O}(\sqrt{N})$ queries, and a result from \citet{bennett1997strengths} shows that Grover's algorithm is optimal up to a multiplicative constant. \citet{boyer1998tight} extend Grover's algorithm to multiple marked items or solutions, further extend it to the case when the number of solutions is unknown and provide a rigorous bounds on the algorithms' error rates.  Finally, \citet{durr1996quantum} use the results of \citet{boyer1998tight} to obtain the minimum value within a discrete ordered set.  Here, we briefly review these advances and provide illustrative simulations.  In Section \ref{sec:gm}, the Gumbel-max trick allows us to extend these results to the problem of sampling from a discrete distribution.

\subsection{Quantum search}\label{sec:qSearch}

\begin{figure}[!t]
	\centering
	\scalebox{0.9}{
\begin{minipage}{1\textwidth}
\begin{algorithm}[H]
	\caption{Quantum search algorithm \citep{grover1996fast}}\label{alg:grover}
	\KwData{An oracle gate $\U_f$ taking $\ket{x}\ket{y}$ to $\ket{x}\ket{y\oplus f(x)}$ for a function $f(x):\{0,\dots,N-1\}\rightarrow \{0,1\}$ that satisfies $f(x_0)=1$ for a single $x_0$; $n+1=\log_2(N)+1$ quantum states initialized to $\ket{0}^{\otimes n} \ket{1}$; an integer $R=\ceil{\pi\sqrt{N}/4}$ denoting the number of Grover iterations.}
	\KwResult{An $n$-bit binary string $x_0$ satisfying $f(x_0)=1$ with error less than $1/N$.}
	$\ket{0}^{\otimes n}\ket{1} \longrightarrow \H^{\otimes n+1}\ket{0}^{\otimes n}\ket{1}=\ket{h}\ket{-}$; \hspace{5em} \Comment{\hspace{0.3em}$\ket{h}=\frac{1}{\sqrt{N}} \sum_{x=0}^{N-1} \ket{x}$.} 
	$\ket{h}\ket{-} \longrightarrow \G^R \ket{h}\ket{-} \approx \ket{x_0}\ket{-}$;    \hspace{3em}         \Comment{$\G= \Big(2\ket{h}\bra{h}-\I \Big) \Big( \I - 2 \ket{x_0}\bra{x_0} \Big)$}
	$\ket{x_0} \longrightarrow x_0$\;
	\Return{$x_0$}\ .
		\vspace{0.5em}
		
\end{algorithm}
\end{minipage}
}
\end{figure}

\begin{figure}[!t]
	\centering
	\includegraphics[width=0.8\linewidth]{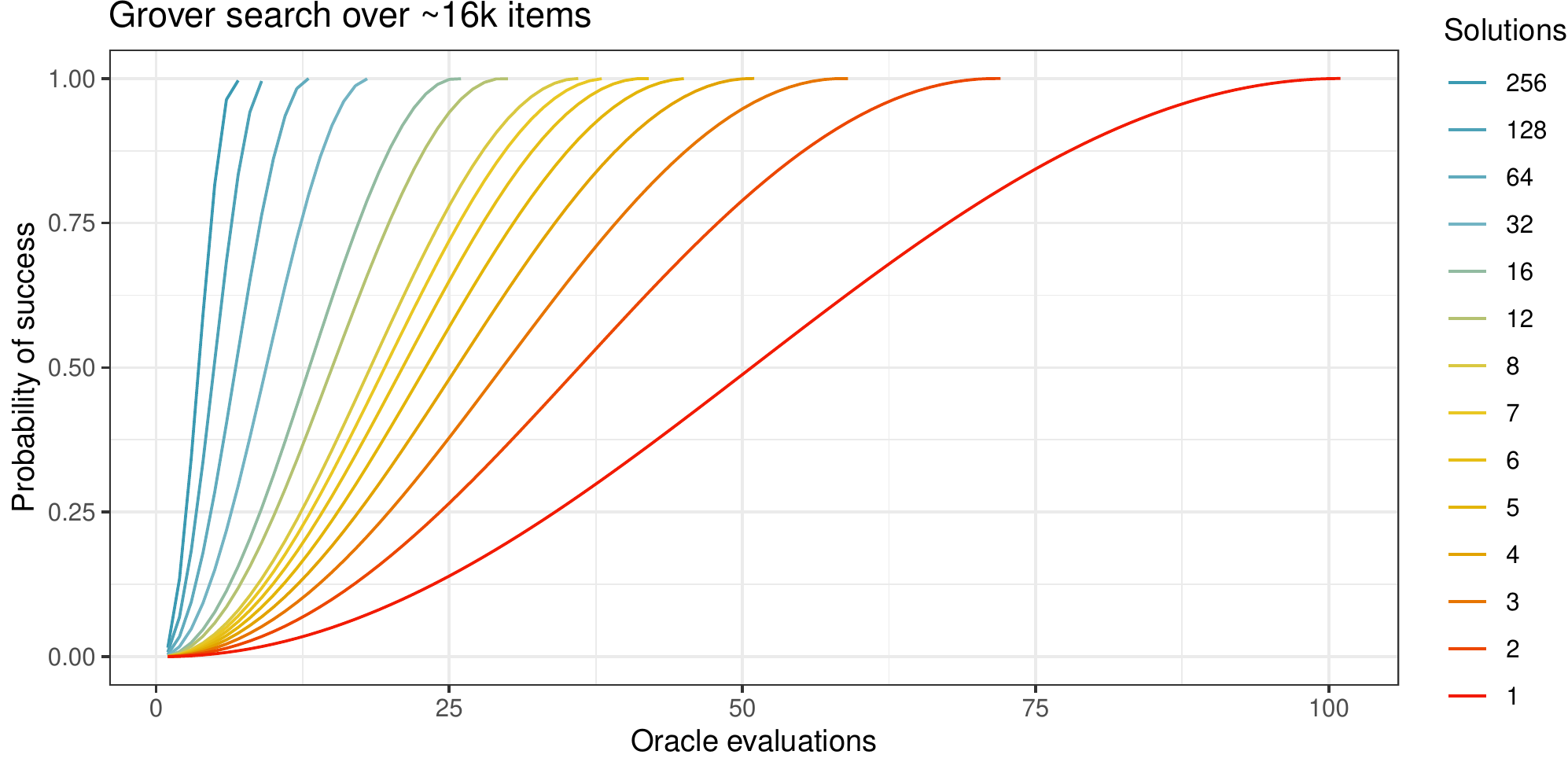}
	\caption{Curves depict the probability that Grover's quantum search algorithm returns a solution as a function of the number of oracle evaluations. For the ranges depicted, the number of iterations required and the number of solutions inversely relate, but increasing the number of iterations past \emph{R} (Equation \ref{eq:R}) can backfire and lead to extremely small probabilities of success.  }\label{fig:grov}
\end{figure}

Given a set of $N$ items and a function $f:\{0,1,\dots,N-1\}\rightarrow \{0,1\}$ that evaluates to $1$ for a single element, \citet{grover1996fast} develops an algorithm that uses quantum parallelism to score quadratic speedups over its classical counterparts. After only $\mathcal{O}(\sqrt{N})$ evaluations of $f(\cdot)$, Grover's algorithm returns the $x\in \{0,1,\dots,N-1\}$ satisfying $f(x) =1$ with high probability.  Compare this to $\order{N}$ requirement for the randomized classical algorithm that must evaluate $f(\cdot)$ over at least $N/2$ items to obtain the same probability of detecting $x$.  The algorithm takes the state $\ket{0}^{\otimes N}\ket{1}$ as input and applies Hadamard gates to each of the individual $N+1$ input qubits.  The resulting state is 
\begin{align*}
\ket{h}\ket{-}= \left(\frac{1}{\sqrt{N}} \sum_{x=0}^{N-1} \ket{x} \right)  \frac{1}{\sqrt{2}}\left(\ket{0} -\ket{1}\right) =  \frac{1}{\sqrt{N}} \sum_{x=0}^{N-1} \ket{x} \ket{-} \, .
\end{align*}
Next, we apply the oracle gate $\U_f: \ket{x}\ket{y} \rightarrow \ket{x}\ket{y\oplus f(x)}$ and note that 
\begin{align*}
\U_f \ket{x} \ket{-} &= \U_f \ket{x} \frac{1}{\sqrt{2}}\left(\ket{0} -\ket{1}\right) 
= \frac{1}{\sqrt{2}} \big(\ket{x} \ket{0\oplus f(x)} - \ket{x}\ket{1\oplus f(x)} \big) \\
&= -1^{f(x)} \ket{x}\ket{-}  \, .
\end{align*}
Thus, $\U_f$ flips the sign for the state $\ket{x_0}$ for which $f(x_0)=1$ but leaves the other states unchanged.  If we suppress the ancillary qubit $\ket{-}$, then $\U_f$ is equivalent to the gate $\U_{x_0}$ defined as $\U_{x_0} \ket{x} = -1^{\delta_{x_0}}\ket{x}$. We may succinctly write this gate as the Householder matrix that reflects vectors about the unique hyperplane through the origin that has $\ket{x_0}$ for a normal vector:
\begin{align*}
\U_{x_0}= \I - 2 \ket{x_0}\bra{x_0} \, .
\end{align*}
The action of this gate on the state $\ket{h}$ takes the form
\begin{align*}
\frac{1}{\sqrt{N}} \sum_{x=0}^{N-1} \ket{x} \longrightarrow \frac{1}{\sqrt{N}} \sum_{x=0}^{N-1} -1^{\delta_{x_0}(x)} \ket{x} \, .
\end{align*}
Next, the algorithm reflects the current state about the hyperplane that has $\ket{h}$ as a normal vector and negates the resulting state:
\begin{align*}
\Big(2\ket{h}\bra{h}-\I \Big) \left( \frac{1}{\sqrt{N}} \sum_{x=0}^{N-1} -1^{\delta_{x_0}(x)} \ket{x}  \right)
&= \frac{1}{\sqrt{N}}\sum_{x} \left( (-1)^{1-\delta_{x_0}(x)} + 2 \frac{(N-2)}{N} \right)\ket{x}\\ &= \frac{(3N-4)}{N^{3/2}} \ket{x_0} + \sum_{x\neq x_0} \frac{(N-4)}{N^{3/2}}  \ket{x} \, .
\end{align*}
The scientist who measures the state at this moment would obtain the desired state $\ket{x_0}$ with a slightly higher probability of $(3N-4)^2/N^3$  than the individual probabilities of $(N-4)^2/N^3$ for the other states. Each additional application of the \emph{Grover iteration}
\begin{align}\label{eq:grov}
\G := - \U_h \U_{x_0} := \Big(2\ket{h}\bra{h}-\I \Big) \Big( \I - 2 \ket{x_0}\bra{x_0} \Big)
\end{align}
increases the probability of obtaining $\ket{x_0}$ at the time of measurement in the computational basis and $R=\ceil{\pi\sqrt{N}/4}$ iterations guarantees a probability of success that is greater than $1-1/N$.

In general, we may use Grover's search algorithm to find a solution when the number of solutions $M$ is greater than 1. While the algorithmic details change little, the number of required Grover iterations
\begin{align}\label{eq:R}
R = \ceil[\Bigg]{\frac{\pi}{4} \sqrt{\frac{N}{M}}}
\end{align}
 and the probability of success after those $R$ iterations \emph{do} change \citep{boyer1998tight}.  When $M$ is much smaller than $N$, the success rate is greater than $1-M/N$, and even for large $M$ the success rate is more than $1/2$. Lower-bounds are useful for establishing mathematical guarantees, but it is also helpful to understand the quality of algorithmic performance as a function of $M$ and $R$. Figure \ref{fig:grov} shows success curves as a function of the number of Grover iterations (or oracle calls) applied.  The search field contains $2^{14} \approx 16$k elements, and the number of solutions $M$ varies. Each curve represents an individual search task.  The algorithm requires more iterations for smaller numbers and fewer iterations for larger numbers of solutions.  The upper bound on error $M/N$ is only an upper bound: the final probability of success for, e.g., $M=256$ is 0.997 compared to the bound of 0.984.
 
\begin{figure}[!tp]
	\centering
	\scalebox{0.9}{
		\begin{minipage}{1\textwidth}
			\begin{algorithm}[H]
	\caption{Quantum exponential searching algorithm \citep{boyer1998tight}}\label{alg:expoSearch}
	\KwData{An oracle gate $\U_f$ taking $\ket{x}\ket{y}$ to $\ket{x}\ket{y\oplus f(x)}$ for a function $f(x):\{0,\dots,N-1\}\rightarrow \{0,1\}$ with unknown number of solutions; $n=\log_2(N)$.}
	\KwResult{If a solution exists, an $n$-bit binary string $x_0$ satisfying $f(x_0)=1$; if no solution exists, the algorithm runs forever.}
	$m \gets 1$\;
	$\gamma \gets 6/5$\;
	success $\gets$ \textcolor{blue}{FALSE}\;
	\While{\emph{success}$\neq$\emph{\textcolor{blue}{TRUE}}}{
		$j \gets Uniform \{0,\cdots,m-1\}$\;
		$\ket{0}^{\otimes n}\ket{1} \longrightarrow \H^{\otimes n+1}\ket{0}^{\otimes n}\ket{1}= \ket{h}\ket{-}$\; 
		$\ket{h}\ket{-} \longrightarrow \G^j \ket{h}\ket{-} = \ket{x}\ket{-}$;\hspace{1em}  \Comment{$j$ Grover iterations from Alg \ref{alg:grover}.}  
		$\ket{x} \longrightarrow x$; \hspace{16em} \Comment{Measure and check.}  
		\uIf{$f(x)=1$}{$x_0 \gets x$\;
			success $\gets$  \textcolor{blue}{TRUE}\;}
		\Else{$m \gets \min\left( \gamma m, \sqrt{N}\right)$;\hspace{4em}  \Comment{Increase $m$ in case of failure.}
	}}
	\Return{$x_0$}\ .
	
	\vspace{0.5em}
\end{algorithm}
\end{minipage}
}
\end{figure}
 
 While Grover's algorithm delivers impressive speedups over classical search algorithms, it has a major weakness.
 Figure \ref{fig:grov} hides the fact that the probability of success for Grover's algorithm is \emph{not} monatonic in the number of iterations.  Running the algorithm for more than $R$ iterations can backfire.  For example, running the algorithm for $\sqrt{2N}$ iterations when $M=1$ results in a probability of success less than $0.095$ \citep{boyer1998tight}.  The non-monotinicity of Grover's algorithm becomes particularly problematic when we do not know the number of solutions $M$. Taking for example $N=2^{20}$, \citet{boyer1998tight} point out that 804 iterations provide an extremely high probability of success when $M=1$ but a one-in-a-million chance of success when $M=4$.  To solve this problem and develop an effective search algorithm when $M$ is unknown, those authors adopt the strategy of focusing on the expected number of iterations before success. In particular, they propose the quantum exponential search algorithm (Algorithm \ref{alg:expoSearch}).  When a solution exists, the algorithm returns a solution with expected total number of Grover iterations bounded above by $\frac{9}{2}\sqrt{N/M}$.  Still better, this upper bound reduces to $\frac{9}{4}\sqrt{N/M}$ for the special case $M\ll N$, and simulations presented in Figure \ref{fig:expo} show that even this bound is large.  Such results come in handy when deciding whether to halt the algorithm's progress if one believes it possible that no solutions exist.  Indeed, this turns out to be useful in the context of quantum minimization (Section \ref{sec:qmin}).

 \begin{figure}[!tp]
 	\centering
 	\includegraphics[width=1\linewidth]{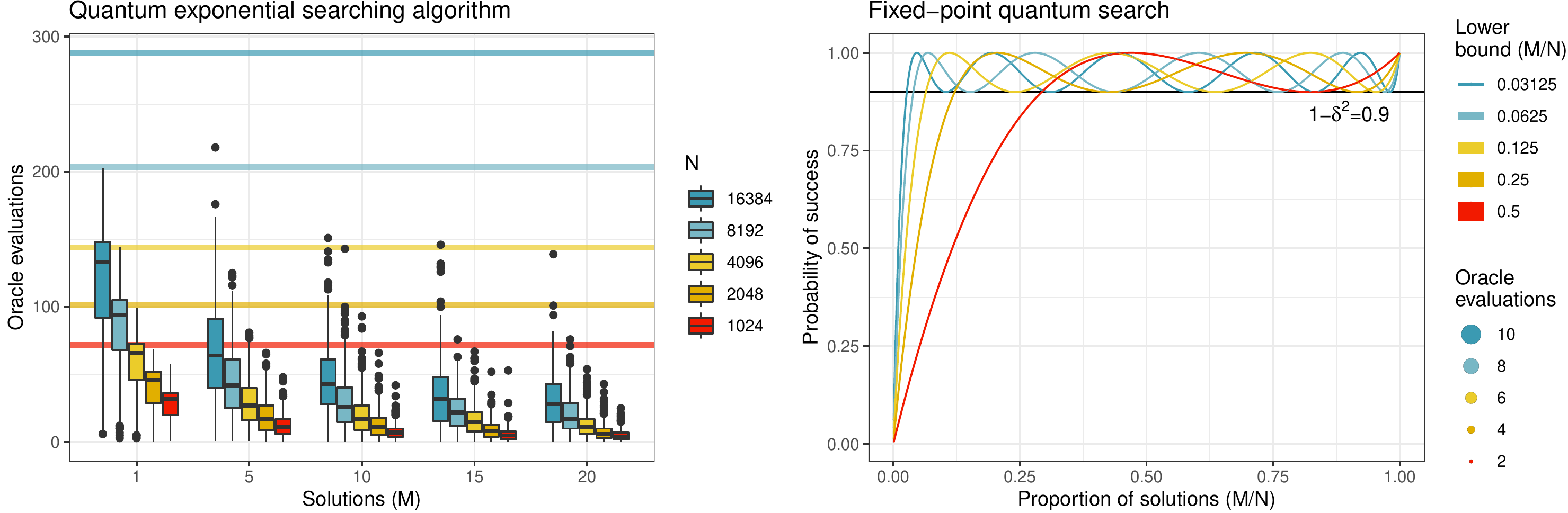}
 	\caption{[Left] Total number of oracle evaluations required by quantum exponential searching algorithm (Algorithm \ref{alg:expoSearch}) for different numbers of solutions $M$ and search set sizes $N$ from 500 independent simulations each.  Horizontal lines at $\frac{9}{4}\sqrt{N}$ represent upper bounds on expected total number of evaluations to obtain a solution for the $M=1$ problem. [Right] Probabilities of success for the fixed-point quantum search algorithm \citep{yoder2014fixed} for different proportions of solutions $\lambda=M/N$ and selecting different lower bounds $w$ on $M/N$ with error tolerance $\delta^2$. }\label{fig:expo}
 \end{figure}

 \begin{figure}[!t]
	\centering
	\includegraphics[width=0.7\linewidth]{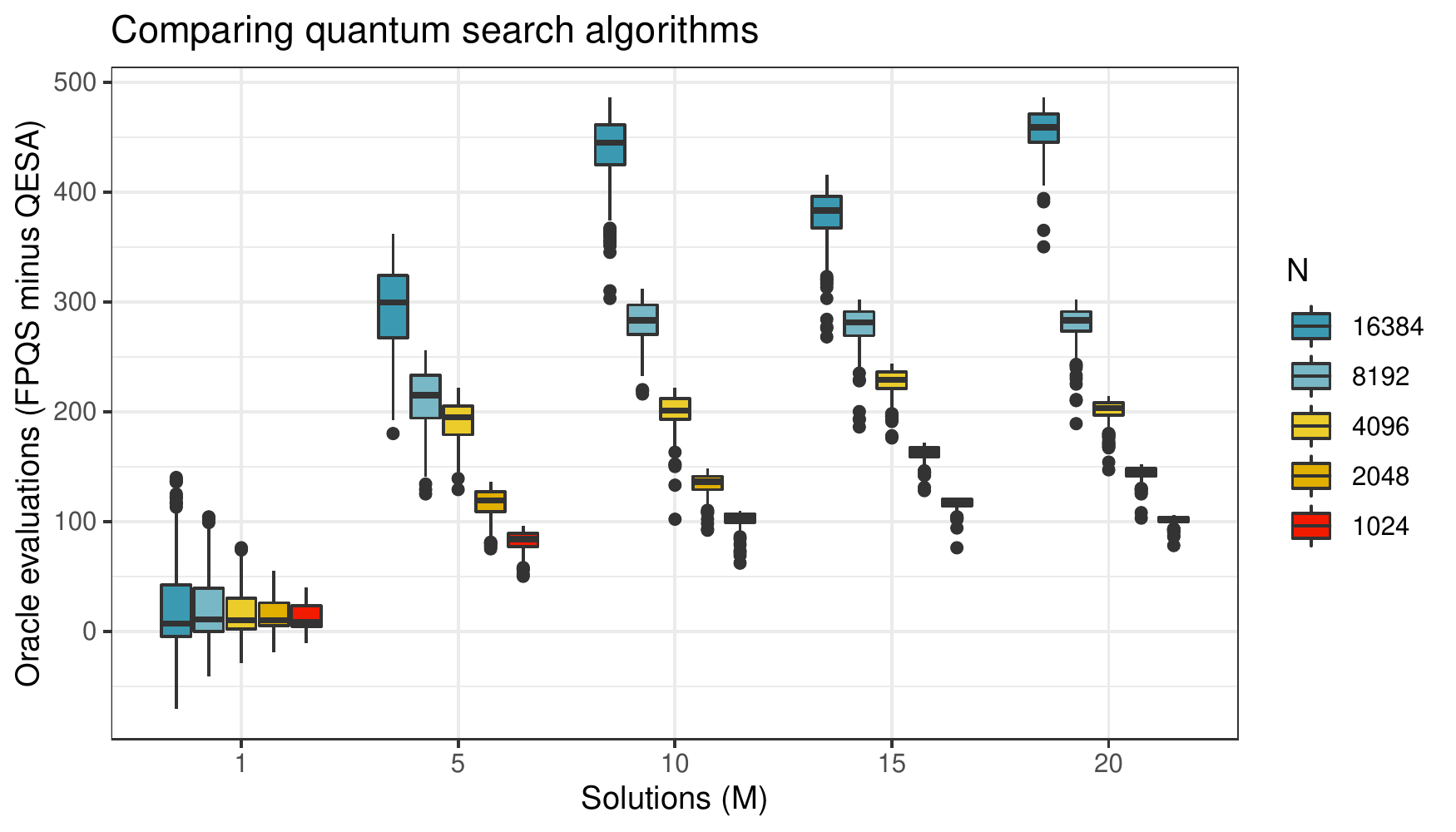}
	\caption{Total number of oracle evaluations required by fixed-point quantum search (FPQS) \citep{yoder2014fixed} minus total required for quantum exponential searching algorithm (QESA) (Algorithm \ref{alg:expoSearch}) for different numbers of solutions $M$ and search set sizes $N$ from 500 independent simulations each.  FPQS underperforms partially due to a miniscule lower bound $w=1/N$ on $\lambda=M/N$, a tuning decision motivated by the potential application within quantum minimization with warm-starting (Section \ref{sec:qmin}).}\label{fig:comp}
\end{figure}

Algorithm \ref{alg:expoSearch} is not the only quantum search algorithm that is useful when the number of solutions $M$ is unknown.  \citet{yoder2014fixed} use \emph{generalized Grover iterations} that extend \eqref{eq:grov} to include general phases $(\alpha_j,\beta_j)$:
\begin{align}\label{eq:grovGen}
	\G(\alpha_j,\beta_j) := - \U_h(\alpha_j) \U_{x_0}(\beta_j) := \Big((1-e^{-i\alpha_j})\ket{h}\bra{h}-\I \Big) \Big( \I - (1-e^{i\beta_j}) \ket{x_0}\bra{x_0} \Big) \, .
\end{align}
This form reduces to the original Grover iteration \eqref{eq:grov} when $\alpha=\pm \pi$ and $\beta=\pm \pi$.
To select general phases, the method first fixes an error tolerance $\delta^2>0$ and a lower bound $w\leq \lambda=M/N$ and then selects an odd upper bound $L$ on the total number of oracle queries ($L-1$) such that 
\begin{align*}
	L \geq \frac{\log(2/\delta)}{\sqrt{w}} \, .
\end{align*}
Finally,  letting $l=(L-1)/2$, one obtains phases $(\alpha_j,\beta_j)$ for $j=1,\dots,l$ that satisfy
\begin{align*}
	\alpha_j = - \beta_{l-j+1} =2 \cot^{-1} \left( \tan(2\pi j /L) \sqrt{1-\gamma^2} \right) \, ,
\end{align*}
 for $\gamma^{-1} =  \cos(\frac{1}{L}\cos^{-1}(\frac{1}{\delta}))$ is the reciprocal of the inverse  $L$th Chebyshev polynomial of the first kind.  Such phases mark the only algorithmic difference between the fixed-point quantum search and the original Grover search (Algorithm \ref{alg:grover}).  The upshot is a search procedure that obtains a guaranteed probability of success $1-\delta^2$ for any $M$ such that there exists an $M_0<M$ that also obtains the same success threshold (Figure \ref{fig:expo}).  On the one hand, this algorithm avoids the exponential quantum search algorithm's need to perform multiple measurements of the quantum state.  On the other hand, the exponential quantum search algorithm requires significantly fewer oracle evaluations when $M$ is small (Figure \ref{fig:comp}), and it turns out that this is precisely the scenario that interests us.
 
 \subsection{Quantum minimization}\label{sec:qmin}

\begin{figure}[!t]
	\centering
	\scalebox{0.9}{
		\begin{minipage}{1\textwidth}
			\begin{algorithm}[H]
 	\vspace{3em}
 	
	\caption{Quantum minimum searching algorithm \citep{durr1996quantum}}\label{alg:min}
	\KwData{A quantum sub-routine capable of evaluating a function $f(\cdot)$ over $\{0,\dots,N-1\}$ with with unique integer values; a maximum error tolerance $\epsilon \in(0,1)$; expected total time to success $m_0=\frac{45}{4}\sqrt{N}+\frac{7}{10}\log_2(N)$.}
	\KwResult{A $\log_2 (N)$-bit binary string $x_0$ satisfying $f(x_0) = \min f$ with probability greater than $1-\epsilon$.}
	$s \gets 0$\;
	$x_0 \gets Uniform\{0,\dots,N-1\}$\;
	\While{$s< m_0/\epsilon$}{
		Prepare initial state $\frac{1}{\sqrt{N}}\sum_x\ket{x}\ket{x_0}$\;
		Mark all items $x$ satisfying $f(x)<f(x_0)$\;
		$s\gets s+ \log_2(N)$\;
		Apply quantum exponential searching algorithm (Alg \ref{alg:expoSearch}); \Comment{$I$ time steps}
		$s\gets s+ I$\;
		Obtain $x'$ by measuring first register\;
		\If{$f(x')<f(x_0)$}{$x_0 \gets x'$}
	}
	\Return{$x_0$} \ .
	\vspace{0.5em}
	
\end{algorithm}
\end{minipage}
}
\end{figure}

Given a function $f(\cdot)$ that maps the discrete set $\{0,\dots,N-1\}$ to the integers, we wish to find the minimizer
\begin{align*}
x_0 = \argmin_{x\in\{0,\dots,N-1\}} f(x) \, .
\end{align*}
 \citet{durr1996quantum} propose a quantum algorithm for finding $x_0$ that iterates between updating a running minimum $F\in f(\{0,\dots,N-1\})$ and applying the quantum exponential search algorithm (Algorithm \ref{alg:expoSearch}) to find an element $x$ such that $f(x)<F$.  Having run these iterations a set number of times, the algorithm returns the minimizer $x_0$ with high probability. 
 To this end, \citet{durr1996quantum} show that their algorithm takes an expected total time less than or equal to $m_0:=\frac{45}{4}\sqrt{N}+\frac{7}{10}\log_2(N)$ to find the minimizer, where marking items with values less than the threshold value (Algorithm \ref{alg:min}) requires $\log_2(N)$ time steps and each Grover iteration within the quantum exponential search algorithm requires one time step. From there, Markov's inequality says
 \begin{align*}
 \mbox{Pr}\left( \mbox{total time to success} \geq \frac{m_0}{\epsilon}\right) \leq \frac{\mbox{E}\left(\mbox{total time to success}\right)}{m_0/\epsilon}\leq \epsilon \, ,
 \end{align*}
 or that we must scale the minimization procedure's time steps by a factor of $1/\epsilon$ to reach a failure rate less than $\epsilon$.  Due to the iterative nature of the algorithm, one might suppose that it is beneficial to start at a value $x_0$ for which $f(x_0)$ is lower relative to other values. This turns out to be the case in theory and in practice, and the benefit of warm-starting is particularly useful in the context of parallel MCMC.

\begin{prop}[Warm-starting]\label{prop:warm}
	Suppose that the quantum minimization algorithm begins with a threshold $F_0$ such that $f(x)<F_0$ for only $K-1>0$ items. Then the expected total number of time steps to find the minimizer is bounded above by
	\begin{align*}
	m_0^K = \left(\frac{5}{4} - \frac{1}{\sqrt{K-1}} \right) 9\sqrt{N} + \frac{7}{10} \log_2(K) \log_2(N) \, ,
	\end{align*}
	and so the following rule relates the warm-started upper bound to the generic upper bound:
	\begin{align*}
	m_0^K = m_0 - 9\sqrt{\frac{N}{K-1}} + \frac{7}{10} \log_2 \left(\frac{K}{N} \right) \log_2(N)\, .
	\end{align*} 
\end{prop}

\begin{proof}
	The proof follows the exact same form as Lemma 2 of \citet{durr1996quantum}. It relies on a theoretical algorithm called the \emph{infinite algorithm} that runs until the minimum is found.  In this case, Lemma 1 of that paper says that the probability that the $r$th lowest value is ever selected when searching among $K$ items is $p(K,r)=1/r$ for $r\leq K$ and 0 otherwise.  For a warm-start at element $K$, the expected total time spent in the exponential search algorithm is
	\begin{align*}
	\sum_{r=2}^N p(K,r) \frac{9}{2} \sqrt{\frac{N}{r-1}} &= \sum_{r=2}^K p(K,r) \frac{9}{2} \sqrt{\frac{N}{r-1}} 
	=  \frac{9}{2}\sqrt{N}\sum_{r=1}^{K-1} \frac{1}{r+1} \frac{1}{\sqrt{r}} \\
	&\leq \frac{9}{2}\sqrt{N}\left( \frac{1}{2}+\sum_{r=2}^{K-1} r^{-3/2} \right) \leq \frac{9}{2}\sqrt{N}\left( \frac{1}{2}+\int_{r=1}^{K-1} r^{-3/2} \dd r \right) \\
	&= \left(\frac{5}{4} - \frac{1}{\sqrt{K-1}} \right) 9\sqrt{N} \, .
	\end{align*}
	An upper bound for the expected total number of time steps preparing the initial state and marking items $f(x)<f(x_0)$ follows in a similar manner.
\end{proof}

Proposition \ref{prop:warm} shows that, e.g., if Algorithm \ref{alg:min} begins at the item with second lowest value, then the expected total time to success is bounded above by $m_0^2=\frac{9}{4}\sqrt{N}+\frac{7}{10}\log_2(N)$, reducing the generic upper bound by $9\sqrt{N}-\frac{7}{10}\log_2\left( \frac{2}{N}\right)\log_2(N)$ time steps.  When $N$ equals 1000, say, the expected total time steps is $m_0^2< 78.2$.  Keeping $N=1000$ but letting the algorithm begin at the third lowest value ($K=3$), the expected total time steps before success is $m_0^3<165.6$. Raising $N$ to 10,000, the two numbers increase to $m_0^2< 234.4$ and $m_0^3<503.4$.  

In practice, the number of time steps before finding the minimum is surprisingly small.  Figure \ref{fig:qMinAlg} corroborates the intuition that it is beneficial to start at a lower ranked element. To generate these results, we use an early stopping rule for the quantum minimization algorithm's exponential searching sub-routine, halting it after $\frac{9}{4}\sqrt{N}$ iterations. Even with this stopping rule in place, the error rate of the quantum minimization algorithm is less than $1\%$. We can increase the early stopping threshold to obtain even lower error rates, but this strategy would seem to be unnecessary in the context of parallel MCMC.  The benefits of warm-starting are useful in the same context, when the current state $\ttheta^{(s)}$ usually inhabits the high-density region of the target distribution but the majority of proposals do not.  

 \begin{figure}[!t]
	\centering
	\includegraphics[width=0.7\linewidth]{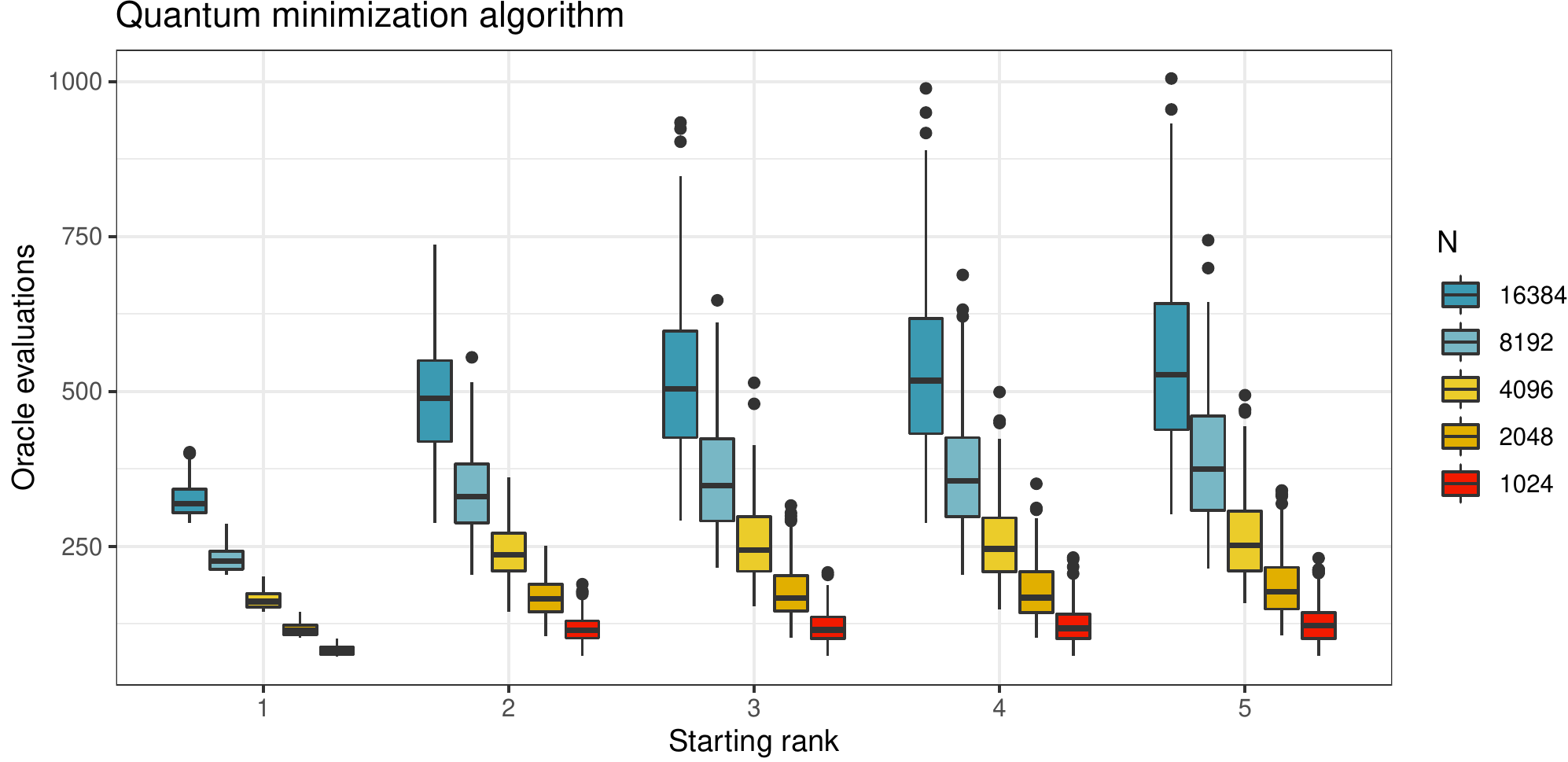}
	\caption{Total number of oracle evaluations required by quantum minimization algorithm \citep{durr1996quantum} for different starting ranks and search set sizes $N$ from 500 independent simulations each.  Here, less than 1\% of the 12,500 instances fail to recover the true minimum on account of early stopping.}\label{fig:qMinAlg}
\end{figure}

\section{Quantum parallel MCMC}\label{sec:qpmcmc}

With the rudiments of quantum minimization in hand, we present a quantum parallel MCMC (QPMCMC). The general structure of the algorithm the same as Algorithm \ref{alg:pMCMC}: it constructs a Markov chain by iterating between generating many proposals and randomly selecting new Markov chain states among these proposals.  Unlike classical single-proposal MCMC, parallel MCMC proposes many states and chooses one according to the generalized acceptance probabilities of \eqref{eq:probs}.
 In general MCMC, evaluation of the target density function $\pi(\cdot)$ is the rate-limiting computational step, becoming particularly onerous in big data scenarios \citep{massive}.  While parallel MCMC benefits from improved mixing as a function of MCMC iterations, the increased computational burden of $P$ target evaluations at each step can lead to less favorable comparisons when accounting for elapsed wall-clock time.
 
Having successfully generated proposals $\ttheta_p,$ $p=1,\dots,P$ and evaluated the corresponding proposal densities $q(\cdot,\cdot)$, we would like to use quantum parallelism and an appropriate oracle gate to efficiently compute the $\pi(\ttheta_p)$s but immediately encounter a Catch-22 when we seek to draw a sample $\ttheta^{(s+1)} \sim$ \emph{Discrete}$(\ppi)$ for $\ppi$ the vector of probabilities $\pi_0,\pi_1,\dots,\pi_P$ defined in \eqref{eq:probs}.  We can use neither a quantum nor a classical circuit to draw the sample! On the one hand, drawing the sample within a quantum circuit would somehow require that all superposed states have access to all the $\pi(\ttheta_p)$s at once to perform the required normalization.  On the other hand, drawing the sample within the classical circuit would require access to all the $\pi(\ttheta_p)$s, but measurement in the computational basis can only return one. 
 
 In light of this dilemma, we propose to use the Gumbel-max trick \citep{papandreou2011perturb} to transform the generalized accept-reject step into a discrete optimization procedure (Section \ref{sec:gm}).  From there, we use quantum minimization to efficiently sample from $Discrete(\ppi)$. Crucially, we get away with quantum parallel evaluation of the target $\pi(\cdot)$ over all superposed states because the Gumbel-max trick does not rely on a normalization step: each superposed state requires no knowledge of any target evaluation other than its own.  

Once one has effectively parallelized the target evaluations $\pi(\ttheta_{p})$, the new computational bottleneck is the calculations required to obtain the $P+1$ proposal density terms $q(\ttheta_{p},\Ttheta_{-p})$, for $p=0,\dots,P$.  Under an independent proposal mechanism with a proposal distribution of the form $q(\ttheta_0,\Ttheta_{-0})=\prod_{p=1}^Pq(\ttheta_0,\ttheta_{p})$, the acceptance probabilities \eqref{eq:probs} require the $\order{P^2}$ evaluation of the $P+1\choose 2$ terms $q(\ttheta_{p},\ttheta_{p'})$.  To avoid such calculations, we follow \citet{holbrook2021generating} and use symmetric proposal distributions for which $q(\ttheta_{p},\Ttheta_{-p}) = q(\ttheta_{p'},\Ttheta_{-p'})$ for $p,p'=0,\dots,P$.

\newcommand{\p}{\mbox{p}}

\begin{figure}[!t]
	\centering
	\scalebox{0.9}{
		\begin{minipage}{1\textwidth}
			\begin{algorithm}[H]
	\caption{The Gumbel-max trick}\label{alg:gm}
	\KwData{A vector of unnormalized log-probabilities $\llambda^*=\log \ppi +\log c$, for $\ppi$ a discrete probability vector with $P+1$ elements.}
	\KwResult{A single sample $\hat{p} \sim$ \emph{Discrete}$(\ppi)$ satisfying $\hat{p}\in\{0,1,\dots,P\}$.}
	\For{$p \in\{0,1,\dots,P\}$}{
		$z_{p} \gets Gumbel(0,1)$\;
		$\alpha^*_{p} \gets \lambda_{p}^* + z_{p}$\;
	}
	$\hat{p} \gets \argmax_{p=0,\dots,P}\alpha^*_{p}$\;
	\Return{$\hat{p}\,$}.
	
	\vspace{0.5em}
\end{algorithm}
\end{minipage}
}
\end{figure}

\subsection{Simplified acceptance probabilities}\label{sec:sym}

\newcommand{\SSigma}{\boldsymbol{\Sigma}}

Evaluation of the $P+1$ proposal density terms $q(\ttheta_{p},\Ttheta_{-p})$ can be computationally burdensome as $P$ grows large.  \citet{holbrook2021generating} proves that two specific proposal mechanisms enforce equality across all $P+1$ terms and thus enable the simplified acceptance probabilities
\begin{align}\label{eq:probsSimp}
	\pi_p = \frac{\pi(\ttheta_p) }{\sum_{p'=0}^P \pi(\ttheta_{p'})} \, , \quad p=0,\dots,P \, .
\end{align}
In Section \ref{sec:continu}, we opt for one of these proposal mechanisms, the centered Gaussian proposal of \citet{tjelmeland2004using}.  This strategy first generates a Gaussian random variable $\bar{\ttheta}$ centered at the current state $\ttheta_{0}$ and then generates $P$ proposals centered at $\bar{\ttheta}$:
\begin{align}\label{eq:p1}
	\ttheta_1,\dots,\ttheta_P \stackrel{\perp}{\sim} Normal_D(\bar{\ttheta},\SSigma) \, , \quad \bar{\ttheta} \sim Normal_D(\ttheta_{0},\SSigma) \, .
\end{align}
\citet{holbrook2021generating} shows that this strategy leads to the higher-order proposal symmetry
\begin{align}\label{eq:symmetry}
	q(\ttheta_0,\Ttheta_{-0}) = q(\ttheta_{1},\Ttheta_{-1}) =\dots =q(\ttheta_{P},\Ttheta_{-P})  
\end{align}
and that the simplified acceptance probabilities \eqref{eq:probsSimp} maintain detailed balance.  In fact, other location scale families also suffice in the manner, but here we focus on Gaussian proposals without loss of generality.  Finally, the simplicial sampler \citep{holbrook2021generating} accomplishes the same simplified acceptance probabilities but incurs an $\mathcal{O}(D^3)$ computational cost.

One may also use a more limited strategy to enforce \eqref{eq:symmetry}. \citet{glatt} show that the sampler using independence proposal $q(\ttheta_p,\Ttheta_{-p})=\prod_{p'\neq p} q(\ttheta_{p'})$, where $q(\cdot)$ is not a function of $\ttheta_p$, results in acceptance probabilities
\begin{align}\label{eq:indepProb}
	\pi_p = \frac{\pi(\ttheta_p) /q(\ttheta_p)}{\sum_{p'=0}^P \pi(\ttheta_{p'})/q(\ttheta_{p'})} \, , \quad p=0,\dots,P \, .
\end{align}
When $\pi(\cdot)$ and $q(\cdot)$ take continuous values on a topologically compact domain or discrete values on a finite domain, one may specify $q(\cdot)$ to be uniform and let \eqref{eq:indepProb} simplify to \eqref{eq:probsSimp}.  This strategy proves useful in Section \ref{sec:discrete}, in which we consider discrete-valued targets over Ising-type lattice models.

\subsection{Continuously-valued targets}\label{sec:continu}

\begin{figure}[!t]
	\centering
	\scalebox{0.9}{
		\begin{minipage}{1\textwidth}
			\begin{algorithm}[H]
	\caption{Quantum parallel MCMC for a continuously-valued target}\label{alg:qpMCMC}
	\KwData{Initial Markov chain state $\ttheta^{(0)}$; total length of Markov chain $S$; total number of proposals per iteration $P$; routine for evaluating target density $\pi(\cdot)$; routines for drawing random samples from a $D$-dimensional multivariate Gaussian $Normal_D(\boldsymbol{\mu},\SSigma)$ and the standard Gumbel distribution $Gumbel(0,1)$.}
	\KwResult{A Markov chain $\ttheta^{(1)}, \dots, \ttheta^{(S)}$.}
	\For{$s \in\{1,\dots,S\}$}{
		$\ttheta_0 \gets \ttheta^{(s-1)}$\;
		$z_{0} \gets Gumbel(0,1)$\;
		$\bar{\ttheta} \gets Normal_D(\ttheta_0,\SSigma)$; \hspace{14em} \Comment{ \hspace{0.3em}     Proposal \eqref{eq:p1}}      
		\For{$p \in \{1,\dots,P\}$}{
			$\ttheta_p \gets Normal_D(\bar{\ttheta},\SSigma)$\;
			$z_{p} \gets Gumbel(0,1)$\;
		\HiLi	$\ket{\theta_p} \gets \ttheta_p$; \hspace{5em} \Comment{      Load proposal onto quantum computer.}      
		}
	\HiLi	$\hat{p} \gets \argmin_{p=0,\dots,P}\Big(f(p):= -\big(z_{p}+ \log\pi(\ttheta_{p}) \big)\Big)$\;  \HiLi \hspace{10em}\Comment{\hspace{0.5em} $\order{\sqrt{P}}$ Alg \ref{alg:min} with warm-start at $p=0$.}
		$\ttheta^{(s)} \gets \ttheta_{\hat{p}}$\;
	}
	\Return{$\ttheta^{(1)}, \dots, \ttheta^{(S)}$}\ .
	
	\vspace{0.5em}
\end{algorithm}
\end{minipage}
}
\end{figure}

 \begin{figure}[!t]
 	\vspace{3em}
	\centering
	\includegraphics[width=0.9\linewidth]{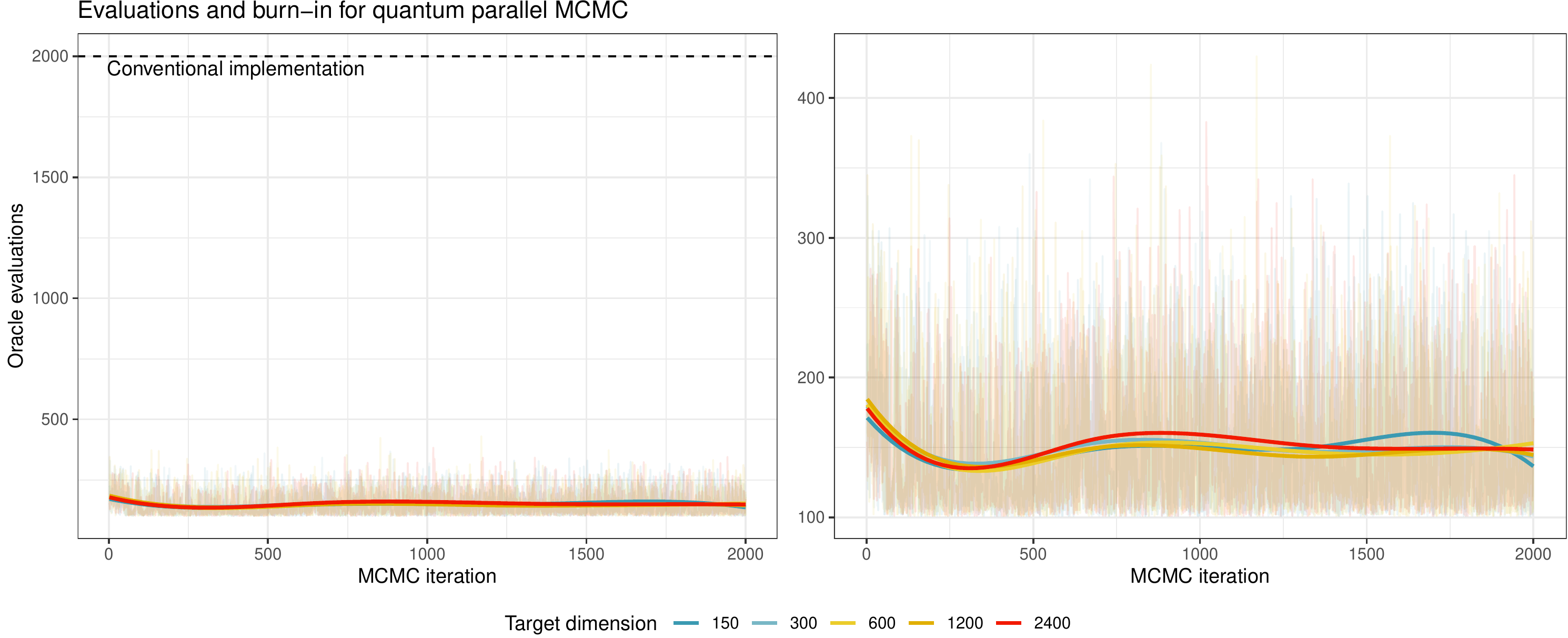}
	\caption{Total number of oracle evaluations required for each of 2,000 quantum parallel MCMC (QPMCMC) iterations for standard multivariate normal targets of five different dimensionalities. Regardless of target dimension, the individual QPMCMC runs require roughly $7$\% of the usual 4 million target evaluations. Over  99.4\% of the 10,000 MCMC iterations across all dimensions successfully sample from the discrete distribution with probabilities of \eqref{eq:probs}. Burn-in iterations require moderately more evaluations because the current state occupies a lower density region and represents a `less good' warm-start.}\label{fig:mcmcIterations}
\end{figure}

 \begin{figure}[ph!]
	\centering
	\includegraphics[width=0.5\linewidth]{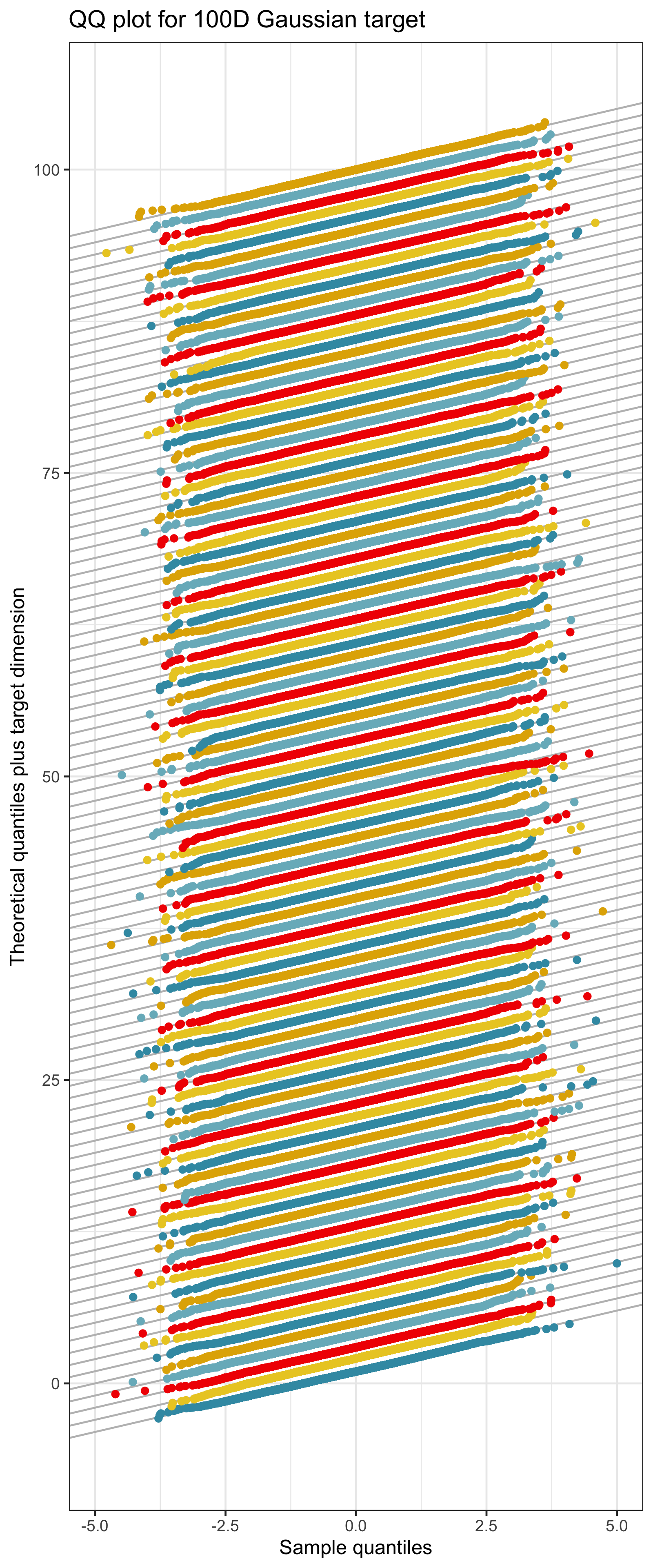}
	\caption{Empirical accuracy of quantum parallel MCMC (QPMCMC) for a 100 dimensional spherical Gaussian target.  We generate 100,000 samples using 2,000 proposals each iteration and remove the first 2,000 samples as burn-in.  The QQ (quantile-quantile) plot shows sample quantiles adhering closely to the theoretical quantiles.  Similar to the independent simulation shown in Figure \ref{fig:mcmcIterations}, here QPMCMC requires less than $7.2$\% of the usual number of target evaluations.}\label{fig:qq}
\end{figure}

Algorithm \ref{alg:qpMCMC} presents the details of QPMCMC for a continuously-valued target.  The algorithm uses conventional (perhaps parallel) computing almost entirely, excluding two lines that are highlighted.  The first of these loads proposals onto the quantum machine, and the second line that calls the quantum minimization algorithm presented in Algorithm \ref{alg:min}. 
This sub-routine seeks to find the minimizer of the function
\begin{align}\label{eq:objective}
f(p)= -\big(z_{p}+ \log\pi(\ttheta_{p}) \big) \quad \mbox{over} \quad p=0,1,\dots, P,
\end{align} 
which combines the numerators of \eqref{eq:probs}, the Gumbel-max trick and a trivial negation.  As discussed in Section \ref{sec:qmin}, the sub-routine requires only $\mathcal{O}(\sqrt{P})$ oracle evaluations and, therefore, only $\mathcal{O}(\sqrt{P})$ evaluations of the target density $\pi(\cdot)$ using a quantum circuit.  In theory, a quantum computer can perform the same computations as a classical computer, but the efficiency of the target evaluations would depend on a number of factors related to the structure of the density itself, the availability of fast quantum random access memory \citep{giovannetti2008quantum} and the ability of large numbers of quantum gates to act in concert with negligible noise \citep{kielpinski2002architecture,erhard2019characterizing}.

In general MCMC, one often calibrates the scaling of a proposal kernel to balance between exploring the target's domain and remaining within high-density regions.  Optimal scaling strategies may lead to a large number of rejected proposals \citep{roberts2001optimal}.  Indeed, \citet{holbrook2021generating} shows that parallel MCMC algorithms are no exception. When using the Gumbel-max trick to sample from proposals, this means that the current state is often quite near optimality, representing a warm-start.  Figure \ref{fig:mcmcIterations} shows how this warm-starting effects the number of oracle calls (and hence target evaluations) within the quantum minimization sub-routine over the course of an MCMC trajectory.  We target multi-dimensional standard normal distributions of differing dimensions using the vector $(100,\dots,100)$ as starting position. We fix the number of Gaussian proposals to be 2,000 and use standard adaptation procedures \citep{rosenthal2011optimal} to target a 50\% acceptance rate while guaranteeing convergence. Although no theoretical results validate this target acceptance rate, the rate is close to empirically optimal rates from \citet{holbrook2021generating}. Across iterations, QPMCMC requires relatively few oracle calls compared to the 2,000 target evaluations of parallel MCMC. We also witness the expected drop in the number of oracle evaluations as the chain approaches the target's high-density region. Remarkably, the algorithm succeeds in sampling from the true discrete distribution in 99.5\% of the MCMC iterations. An independent simulation obtains similar results, requiring roughly 7\% of the usual number of target evaluations. Figure \ref{fig:qq} shows a quantile-quantile plot for all 100 dimensions of a multi-dimensional standard normal distribution. We see no appreciable impact from the rare failure of the quantum minimization sub-routine.

%\begin{table}[ht]\centering
%	\begin{tabular}{llll}  \toprule 
%	Proposals	& Effective samples (ES) & Evaluations per ES & Speedup \\   \midrule
%	1 (MH)	 & 21.35  (11.74, 50.65) & 429.81  (180.15, 766.89) & 1.00  (1.00, 1.00) \\  
%		 50 & 736.25 (418.73, 1072.35) & 259.54  (177.72, 475.16) & 1.66  (1.01, 1.61) \\  
%		 100 & 1842.71 (1604.07, 2097.74) & 146.49  (128.76, 168.39) & 2.93 (1.40, 4.55) \\  
%		 150 & 2287.03  (1734.34, 2444.05) & 146.39 (137.36, 193.16) & 2.94 (1.31, 3.97) \\ 
%		  200 & 2662.89  (2547.71, 2844.14) & 145.92 (136.37, 152.74) & 2.95  (1.32, 5.02) \\ 
%		   250 & 2971.59  (2439.45, 3208.46) & 147.77  (136.53, 180.95) & 2.91  (1.32, 4.24) \\  
%		  \bottomrule\end{tabular}
%\end{table}

\sloppy
Finally, we compare convergence speed for QPMCMC using 1,000, 5,000 and 10,000 proposals to target a massively multimodal distribution: a mixture of 1,000 2-dimensional, isotropic Gaussians with standard deviations equal to 1 and means equal to $10\times(\boldsymbol{0}, \boldsymbol{1} , \dots , \boldsymbol{999})$. This target is particularly challenging because the distance between modes is significantly larger than standard deviations. In 5 independent simulations, we run each algorithm until it achieves and effective sample size of 100.  Table \ref{tab:res} contains MCMC iterations and target evaluations required to achieve this effective sample size as well as relative speedups over sequential implementations and relative improvements over the 1,000 proposal sampler.

\let\oldtabular\tabular
\let\endoldtabular\endtabular
\renewenvironment{tabular}{\rowcolors{2}{white}{trevorblue!15}\oldtabular}{\endoldtabular}

\newcommand{\ra}[1]{\renewcommand{\arraystretch}{#1}}
\ra{1.2}
\begin{table}[t!]
	\centering
	\resizebox{\textwidth}{!}{
		\begin{tabular}{@{}lllll@{}}\toprule
			Proposals & MCMC iterations  & Target evaluations & Speedup  & Efficiency gain  \\   \midrule
			1,000 & 249,398 (200,998,  311,998) & 24,988,963 (20,149,132, 31,265,011) & 9.98  (9.98, 9.98) & 1  \\   
			5,000 & 14,398 (12,998, 16,998) & 3,314,560 (2,989,418,  3,916,281) & 21.72  (21.70, 21.74) & 7.58  (6.25, 9.71) \\   
			10,000 & 5,998  (4,998,  6,998) & 1,993,484  (1,662,592,  2,330,842) & 30  (29.96, 30.26) & 12.87  (8.64, 18.80) \\    \bottomrule
		\end{tabular}
	}
	\caption{Racing to a minimum (between the two dimensions) of 100 effective samples for a target with 1,000 disjoint modes. `Speedup' is ratio between target evaluations required for sequential and quantum implementations. `Efficiency gain' is ratio between target evaluations required for 1,000 proposal and 5,000/10,000 proposal implementations. We report means (minima, maxima) across 5 independent runs.}\label{tab:res}
\end{table}

\subsection{Discrete-valued targets}\label{sec:discrete}

We wish to sample from a target distribution defined over a discrete set $\{\ttheta_\alpha\}_{\alpha \in \mathcal{A}}$, for $\mathcal{A}$ some finite or countably infinite set of indices.  \citet{glatt} establish the broader measure theoretic foundations for parallel MCMC procedures that use selection probabilities \eqref{eq:probs} to maintain detailed balance. In particular, when the target distribution has probability mass function $\pi(\cdot)$ defined with respect to the counting measure on the power set $2^{\mathcal{A}}$, then detailed balance results in the kernel $P(\cdot,\cdot)$  satisfying
\begin{align*}
	\pi(\alpha) = \sum_{{\alpha'}} \pi(\alpha') P(\alpha',\alpha) \, , \quad \forall \alpha,\alpha' \in \mathcal{A} \, .
\end{align*}
Sections \ref{sec:ising} and \ref{sec:bayesIm} consider targets defined over finite sets $\mathcal{A}$ and make use of the uniform independence proposal scheme described in Section \ref{sec:sym}, thereby facilitating simplified acceptance probabilities \eqref{eq:probsSimp}.  This scheme proposes single-flip updates to the current Markov chain state $\ttheta_0$, but it is worth noting that other multiple-flip schemes would also be amenable to QPMCMC.

\subsubsection{Ising model on a 2-dimensional lattice}\label{sec:ising}

 \begin{figure}[!t]
	\centering
	\includegraphics[width=0.7\linewidth]{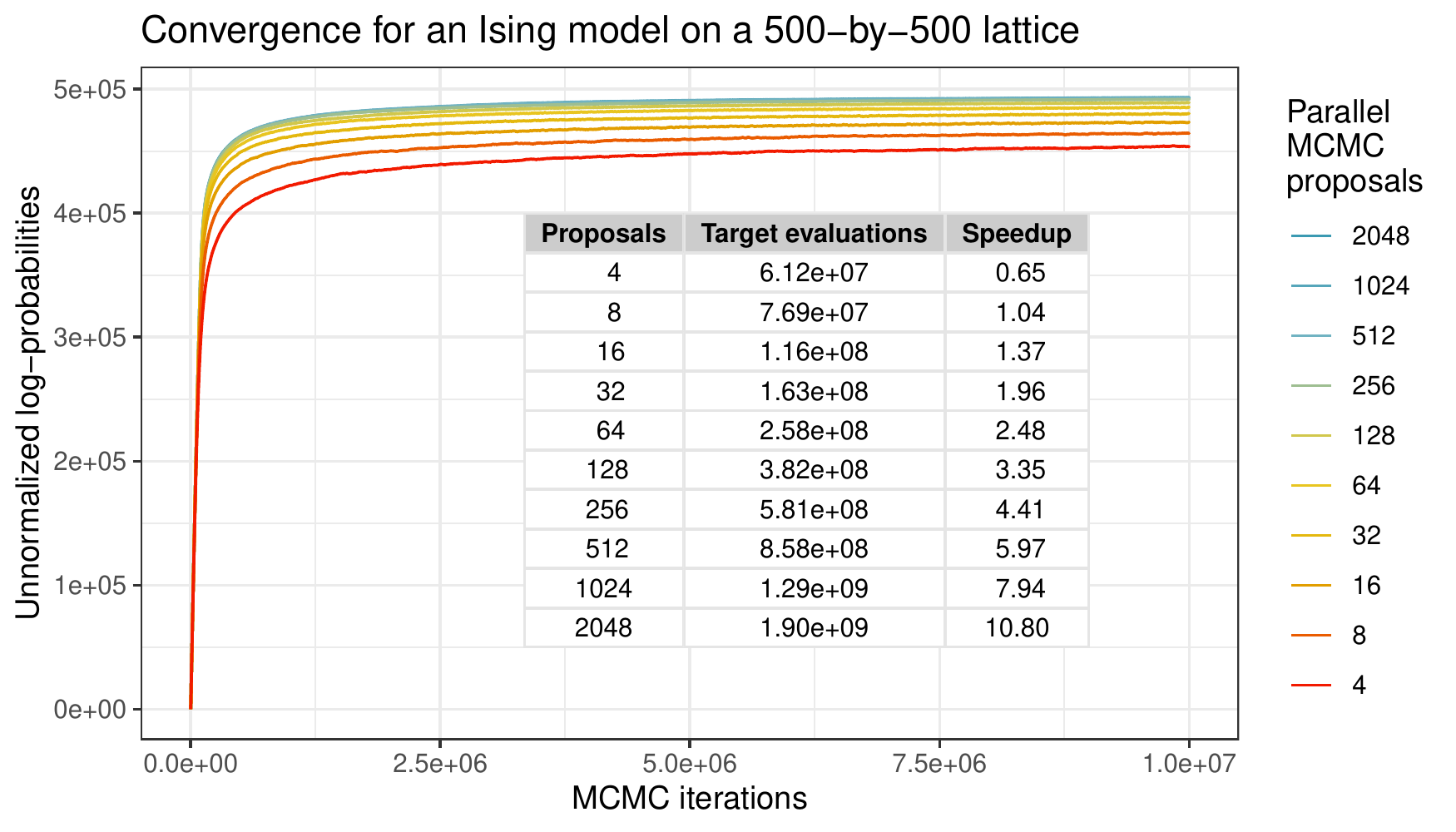}
	\caption{Convergence of log-posterior for a parallel MCMC sampler targeting an Ising model with uniform interaction $\rho=1$ and no external field. All chains begin at minimal probability `checkerboard' state. Larger proposal counts allow discovery of higher probability configurations.}\label{fig:ising2d}
\end{figure}

Here, we are interested in an Ising-type lattice model over configurations $\ttheta=(\theta_1,\dots,\theta_D)$ consisting of $D$ individual spins $\theta_d \in \{-1,1\}$. In terms of the preceding section, we have $\mathcal{A}=\{-1,1\}^D$ and $|\mathcal{A}|=2^D$. In particular, we consider targets of the form
\begin{align}\label{eq:ising}
	\pi (\ttheta | \rho) &\propto \exp \left( \rho \sum_{(d,d')\in \mathcal{E}} \theta_d \theta_{d'} \right)  \, ,
\end{align}
where $\rho>0$ is the interaction term and $\mathcal{E}$ is the lattice edge set.  We specify a two-dimensional 500-by-500 lattice with nearest neighbors connections and fix $\rho=1$.  Since this latter setting encourages neighboring spins to equal each other, we begin our QPMCMC trajectories at the lowest-probability `checkerboard' state for an initial configuration.  At each iteration, we generate collections of proposal states by uniformly flipping $P$ individual spins $\theta_p$, $p \in \{1,\dots,P\}$, and obtaining each proposal state $\ttheta_p$ by updating the current state $\ttheta_0$ at the single location corresponding to $\theta_p$. Figure \ref{fig:ising2d} shows results from 10 independent runs using $P\in \{4, 8, 16, \dots, 2048\}$ proposals.  Trace plots of unnormalized log-probabilities indicate that higher proposal counts enable discovery of higher probability configurations.  Interestingly, QPMCMC appears to be particularly beneficial in this large $P$ context, requiring less than one-tenth the target evaluations required using conventional computing when $P=2048$.  

\subsubsection{Bayesian image analysis}\label{sec:bayesIm}

 \begin{figure}[!t]
	\centering
	\includegraphics[width=0.9\linewidth]{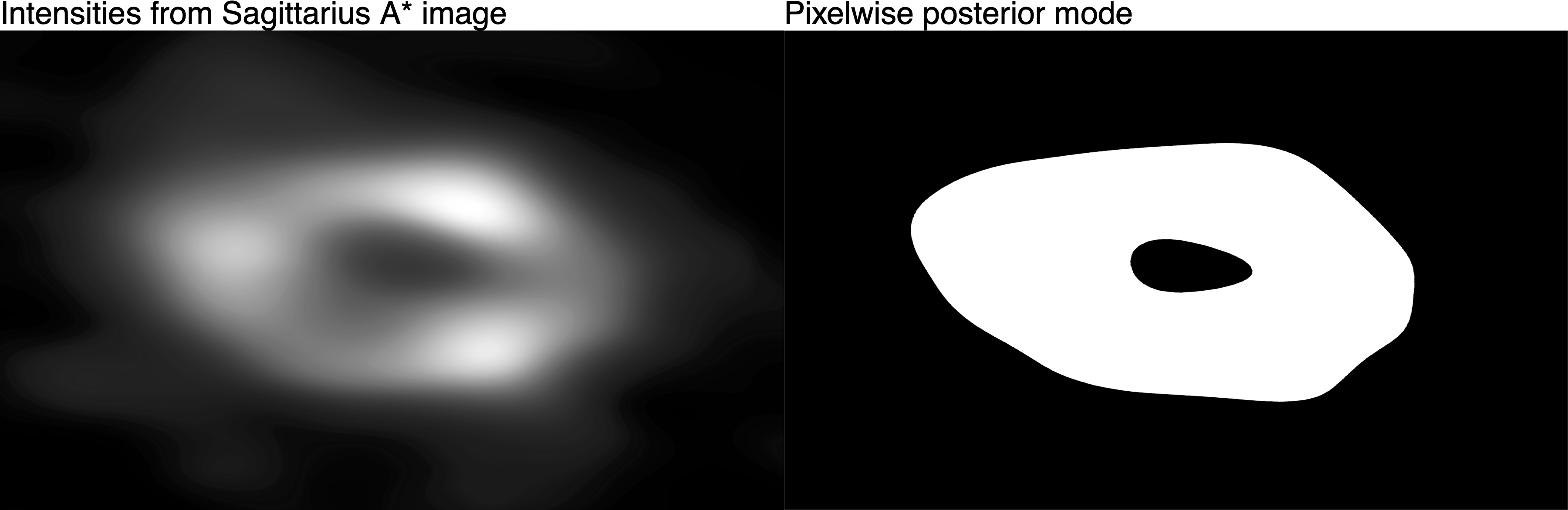}
	\caption{On the left is a 4,076-by-4,076 intensity map of the shadow of supermassive black hole Sagittarius A*  \citep{akiyama2022first}.  On the right is the pixelwise posterior mode of a Bayesian image classification model fit to intensity data.  Within a Metropolis-in-Gibbs routine, quantum parallel MCMC using 1,024 proposals requires less that one-tenth the posterior evaluations required by conventional parallel MCMC.}\label{fig:blackHole}
\end{figure}

We apply a Bayesian image classification model to an intensity map (Figure \ref{fig:blackHole}) of the newly imaged supermassive black hole, Sagittarius A*, located at the Galactic Center of the Milky Way \citep{akiyama2022first}.  Whereas one cannot see the black hole itself, one may see the shadow of the black hole cast by the hot, swirling cloud of gas surrounding it.  We extract the black hole from the surrounding cloud by modeling the intensity at each of the $D=$4,076$^2=$16,613,776 pixels as belonging to a mixture of two truncated normals with values $y_{d}$ restricted to the intensity range $[0,255]$.  Namely, we follow \citet{hurn1997difficulties} and specify the latent variable model
\begin{align*}
	y_d |(\mu_\ell, \sigma^2, \theta_d) &\stackrel{ind}{\sim} \mbox{Normal}(\mu_\ell,\sigma^2 ) \, ,\quad y_d \in [0,255] \, , \quad \theta_d=\ell \,,\quad d\in \{1,\dots,D\}\, ,\\
	\mu_\ell &\stackrel{iid}{\sim} \mbox{Uniform} (0,255) \,  , \quad \ell \in \{-1,1\} \, , \\
	\frac{1}{\sigma^2} &\sim  \mbox{Gamma}\left(\frac{1}{2}, \frac{1}{2} \right)  \, ,
\end{align*}
where $\ttheta=(\theta_1,\dots,\theta_D)$ share for a prior the Ising model \eqref{eq:ising} with edges joining neighboring pixels and interaction $\rho=1.2$. 

We use a QPMCMC-within-Gibbs scheme to infer the join posterior of $\ttheta$ and the three mixture model parameters.  For the former, we use the same QPMCMC scheme as in Section \ref{sec:ising} with 1,024 proposals at each iteration. For the latter, we use the closed-form updates presented in \citet{hurn1997difficulties}.  We run this scheme for 20 million iterations, discarding the first 10 million as burnin.  We thin the remaining sample at a ratio of 1 to 40,000 for the latent memberships $\ttheta$ and 1 to 4,000 for the three parameters $\mu_{-1}$, $\mu_1$ and $\sigma^2$.  Using the \textsc{R} package \textsc{coda} \citep{coda}, we calculate effective sample sizes of the log-posterior (257.1), $\mu_{-1}$ (1,578.7), $\mu_1$ (257.6) and $\sigma^2$ (2,500.0), suggesting adequate convergence.   Figure \ref{fig:blackHole} shows both the intensity data and the pixelwise posterior mode of the latent membership vector $\ttheta$.  The QPMCMC requires only 1,977,553,608 target evaluations compared to the 1,024 $\times$ 20,000,000 $=$ 2.048$\times10^{10}$ evaluations required for the analogous parallel MCMC scheme implemented on a conventional computer, representing a 10.36-fold speedup.

\section{Discussion}\label{sec:disc}

We have shown that parallel MCMC enjoys quadratic speedups by combining quantum minimization with the Gumbel-max trick.  Within a QPMCMC iteration, the current state represents a warm-start for the sampling-turned-optimization problem, leading to increased efficiencies for the quantum minimization algorithm.   Moreover, combining this approach with the Tjelmeland correction \citep{holbrook2021generating} results in a total complexity reduction from $\mathcal{O}(P^2)$ to $\mathcal{O}(\sqrt{P})$.  Preliminary evidence suggests that our strategy may find use when target distributions exhibit extreme multimodality.  While the algorithm still must construct long Markov chains to reach equilibrium, generating each individual Markov chain state requires significantly fewer target evaluations.

There are major technical barriers to the practical implementation of QPMCMC.  The framework, like other quantum machine learning (QML) schemes, requires on-loading and off-loading classical data to and from a quantum machine.  In the seminal review of QML, \citet{biamonte2017quantum} discuss what they call the `input problem'.  \citet{cortese2018loading} present a general purpose quantum circuit for loading $B$ classical bits into a quantum data structure comprising $\log_2 (B)$ qubits with circuit depth $\order{\log (B)}$.  For continuous targets, QPMCMC requires reading $\order{P}$ real vectors $\ttheta_p\in \mathbb{R}^D$ onto the quantum machine at each MCMC iteration.  If $b$ is the number of bits used to represent a single real value, then reading $B=\order{PDb}$ classical bits into the quantum machine requires time $\order{\log (PDb)}$.  This is true whether one opts for fixed-point \citep{jordan2005fast} or floating-point \citep{haener2018quantum} representations for real values.  The outlook can be better for discrete distributions.  For example, the QPMCMC scheme presented in Section \ref{sec:disc} only requires loading the entire configuration state once prior to sampling.  A $D$-spin Ising model requires $D$ classical bits to encode, and these bits load onto a quantum machine in time $\order{\log(D)}$.  Once one has encoded the entire system state, each QPMCMC iteration only requires loading the addresses of the $\order{P}$ proposal states. If one uses $b$ bits to encode each address, then the total time required to load data onto the quantum machine is $\order{\log(Pb)}$ for each QPMCMC iteration.  On the other hand, the speedup for discrete targets assumes the ability to hold an entire configuration in QRAM.
Conveniently, the `output problem' is less of an issue for QPMCMC, as only a single integer $\hat{p} \in \{0,\dots,P\}$ need be extracted within Algorithm \ref{alg:min}.

This work leads to three interesting questions.   First, what is the status of inference obtained by QPMCMC?  The QPMCMC selection step relies on quantum minimization, an algorithm that only achieves success with probability $1-\epsilon$.  While empirical studies suggest that this error induces little bias, it would be helpful to use this $\epsilon$ to bound the distance between the target distribution and the distribution generated by QPMCMC.  Such theoretical efforts would need to extend recent formalizations of the multiproposal based parallel MCMC paradigm \citep{glatt} to account for biased MCMC kernels.

Second, can QPMCMC be combined with established quantum algorithms that make use of `quantum walks' on graphs to sample from discrete target distributions?  \citet{szegedy2004quantum} presents a quantum analog to classical ergodic reversible Markov chains and shows that such quantum walks sometimes provide quadratic speedups over classical Markov chains.  \citet{szegedy2004quantum} also points out that Grover's search, a key component within QPMCMC, may be interpreted as performing just such a quantum walk on a complete graph.  \citet{wocjan2008speedup} accelerate the quantum walk by using ancillary Markov chains to improve mixing and apply their method to simulated annealing.  Given a quantum algorithm $\mathbb{A}$ for producing a discrete random sample with variance $\sigma^2$, \citet{montanaro} develops a quantum algorithm for estimating the mean of algorithm $\mathbb{A}$'s output with error $\epsilon$ by running algorithm $\mathbb{A}$ a mere $\widetilde{\mathcal{O}}(\sigma/\epsilon)$ times, where $\widetilde{\mathcal{O}}$ hides polylogarithmic terms. Importantly, this quadratic quantum speedup over classical Monte Carlo applies for quantum algorithms $\mathbb{A}$ that only feature a single measurement such as certain simple quantum walk algorithms. Unfortunately, this assumption fails for quantizations of Metropolis-Hastings, in general, and QPMCMC, in particular. More promising for QPMCMC, \citet{lemieux2020efficient} develop a quantum circuit that applies Metropolis-Hastings to the Ising model without the need for oracle calls.  An interesting question is whether similar non-oracular quantum circuits exist for the basic parallel MCMC backbone to QPMCMC.  In general, however, comparison between QPMCMC and other quantum Monte Carlo techniques is challenging because the foregoing literature (\textcolor{red}{a}) largely focuses on MCMC as a tool for discrete optimization, with algorithms that only ever return a single Monte Carlo sample or function thereof, and (\textcolor{red}{b}) restricts itself to a handfull of simple, stylized and discrete target distributions.  On the other hand, QPMCMC is a general inferential framework for sampling from general discrete and continuous distributions alike.

Third, there are other models and algorithms in statistics, machine learning and statistical mechanics that require sampling from potentially costly discrete distributions.   Can one adapt  our quantum Gumbel-max trick to these?  Approximate Bayesian computation \citep{csillery2010approximate} extends probabilistic inference to contexts within which one only has access to a complex, perhaps computationally intensive, forward model.  Having sampled model parameters from the prior, one accepts or rejects these parameter values based on the discrepancy between the observed and simulated data.  If one succeeds in embedding forward model dynamics within a quantum circuit, then one may plausibly select from many parameter values using our quantum Gumbel-max trick.  The trick may also find use within sequential Monte Carlo \citep{doucet2001sequential}.  For example, \citet{berzuini2001resample} present a sequential importance resampling algorithm that uses MCMC-type moves to encourage particle diversity and avoid the need for bootstrap resampling.  Multiproposals accelerated by the quantum Gumbel-max trick could add speed and robustness to such MCMC-resampling.

Large-scale, practical quantum computing is still a long way off, but quantum algorithms are ripe for mainstream computational statistics.

\section*{Acknowledgments}

This work was supported by grants NIH K25 AI153816, NSF DMS 2152774 and NSF DMS 2236854.

\appendix

\section{Limited introduction to quantum computing}\label{sec:shortIntro}

Quantum computers perform operations on unit-length vectors of complex data called quantum bits or \emph{qubits}.  One may write any qubit $\ppsi$ as a linear combination of the \emph{computational basis states} $\ket{0}$ and $\ket{1}$.  In symbols,
\begin{align*}
	\ket{\psi} = \alpha \ket{0} + \beta \ket{1} \quad \mbox{for} \quad \alpha, \, \beta \in \mathbb{C} \quad \mbox{and} \quad  |\alpha|^2 + |\beta|^2 = 1\, .
\end{align*}
This formula uses \emph{Dirac} or \emph{bra-ket} notation and obscures ideas that are commonplace in the realm of applied statistics. We make the unit-vector specification of $\ket{\psi}$ clear by writing
\begin{align*}
	\ket{0}=\begin{bmatrix}
		1 \\ 0
	\end{bmatrix} \, ,  \quad \ket{1}=\begin{bmatrix}
		0 \\ 1
	\end{bmatrix} \quad \mbox{or} \quad \ket{\psi} = \alpha \begin{bmatrix}
		1 \\ 0
	\end{bmatrix} + \beta  \begin{bmatrix}
		0 \\ 1
	\end{bmatrix}\, .
\end{align*} 
The full machinery of linear algebra is also available within this notation. The conjugate transpose of $\ket{\psi}$ is $\bra{\psi}$. The inner product of $\ket{\psi}$ and $\ket{\phi}$ is $\braket{\phi|\psi}$. The outer product is $\ket{\psi}\bra{\phi}$. Naturally, we can write $\ppsi$ as a linear combination of any other such basis elements. Consider instead the vectors
\begin{align*}
	\ket{+} = \frac{1}{\sqrt{2}} \ket{0} + \frac{1}{\sqrt{2}} \ket{1} \quad \mbox{and} \quad \ket{-} = \frac{1}{\sqrt{2}} \ket{0} - \frac{1}{\sqrt{2}} \ket{1} \, .
\end{align*}
A little algebra shows that $\ket{+}$ and $\ket{-}$ are indeed unit-length and orthogonal to each other. A little more algebra reveals that, with respect to this basis, $\ppsi$ has coefficients $\alpha'$ and $\beta'$ given by $(\alpha + \beta)/\sqrt{2}$ and $(\alpha - \beta)/\sqrt{2}$, respectively. A last bit of algebra shows that this representation is consistent with $\ppsi$'s unit length.

But linear algebra is not the only aspect of quantum computing that should come easily to applied statisticians. In addition to thinking of $\ppsi$ as a vector, it is also useful to think of $\ppsi$ as a (discrete) probability distribution over the computational basis states $\ket{0}$ and $\ket{1}$. The constraints on coefficients $\alpha$ and $\beta$ mean that we can think of $|\alpha|^2$ and $|\beta|^2$ as probabilities that $\ppsi$ is in state $\ket{0}$ or $\ket{1}$, respectively.  Accordingly, $\ket{+}$ and $\ket{-}$ encode uniform distributions over the computational basis states.  In the parlance of quantum mechanics, $\alpha$, $\beta$ and $\pm1/\sqrt{2}$ are all \emph{probability amplitudes}, and $\ppsi$, $\ket{+}$ and $\ket{-}$ are all \emph{superpositions} of the computational basis states. Quantum \emph{measurement} of $\ppsi$ results in $\ket{0}$ with probability $|\alpha|^2$, but---in the following---we can safely think of this physical procedure as drawing a single discrete sample from $\ppsi$'s implied probability distribution.

Quantum logic gates perform linear operations on qubits like $\ppsi$ and take the form of unitary matrices $\U$ satisfying $\U^\dagger \U=\I$ for $\U^\dagger$ the conjugate transpose of $\U$. One terrifically important single-qubit quantum gate is the \emph{Hadamard gate}
\begin{align*}
	\H = \frac{1}{\sqrt{2}} \left[\begin{array}{rr}
		1 & 1 \\
		1 & -1 \end{array}\right] \, .
\end{align*}
One may verify that $\H$ is indeed unitary and that its action maps $\ket{0}$ to $\ket{+}$ and $\ket{1}$ to $\ket{-}$.  In fact, the reverse is also true on account of the symmetry of $\H$.  The Hadamard gate thus takes the computational basis states in and out of superposition, facilitating a phenomenon called \emph{quantum parallelism}. Given a function $f: \{0,1\} \rightarrow \{0,1\}$, consider the two-qubit quantum \emph{oracle} gate $\U_f$ which takes $\ket{x}\ket{y}$ as input and returns $\ket{x}\ket{y\oplus f(x)}$ as output, where $\oplus$ denotes addition modulo 2.  Importantly, the output simplifies to $\ket{x}\ket{f(x)}$ for $y=0$. We now consider a quantum circuit that acts on two qubits by first applying the Hadamard transform $\H$ to the first qubit and then applying the oracle gate $\U_f$ to both. Using the state $\ket{0}\ket{0}$ as input, we have
\begin{align}\label{eq:quantparallel}
	\ket{0}\ket{0} \longrightarrow \frac{1}{\sqrt{2}} \ket{0}\ket{0} + \frac{1}{\sqrt{2}} \ket{1}\ket{0}  \longrightarrow \frac{1}{\sqrt{2}} \ket{0}\ket{f(0)} + \frac{1}{\sqrt{2}} \ket{1}\ket{f(1)} \, .
\end{align}
The quantum circuit evaluates $f(\cdot)$ simultaneously over both inputs!  Unfortunately, the scientist who implements this circuit cannot directly access $\U_f$'s output, and measurement will provide only $\ket{0}\ket{f(0)}$ or $\ket{1}\ket{f(1)}$ with probability $1/2$ each.  Unlocking the potential of quantum parallelism requires more ingenuity.

Uncountably many single- and two-qubit quantum gates exist, but the real power of quantum computing stems from the development of complex quantum gates that act on multiple qubits simultaneously.   To access this power, we need one more tool that also appears in the statistics literature.  The \emph{Kronecker} or \emph{tensor} product between an $L$-by-$M$ matrix $\A$ and an $N$-by-$O$ matrix $\B$ is the $LN$-by-$MO$ matrix
\begin{align}\label{eq:kron}
	\A \otimes \B = \begin{bmatrix}
		A_{11} \B & \dots & A_{1M} \B \\
		\vdots & \ddots & \vdots \\
		A_{L1}\B & \cdots & A_{LM}  \B
	\end{bmatrix} \, .
\end{align} 
In statistics, the Kronecker product features in the definition of a matrix normal distribution and is sometimes helpful when specifying the kernel function of a multivariate Gaussian process \citep{werner2008estimation}. Here, the product is equivalent to the parallel action of individual quantum gates on individual qubits.  Simple application of Formula \eqref{eq:kron} shows that
\begin{align}\label{eq:twoqubits}
	\H^{\otimes 2} =\H \otimes \H = \frac{1}{2} \left[\begin{array}{rrrr}
		1 & 1 & 1 & 1 \\
		1 &  -1& 1 & -1 \\
		1 & 1 & -1 & -1 \\
		1 & -1 & -1 & 1
	\end{array}\right] \quad \mbox{and} \quad \ket{0}\otimes\ket{0}=\begin{bmatrix}
		1 \\ 0
	\end{bmatrix} \otimes \begin{bmatrix}
		1 \\ 0
	\end{bmatrix} = \begin{bmatrix}
		1 \\ 0 \\ 0 \\ 0
	\end{bmatrix} \, .
\end{align} 
We may also denote the left product $\H^{\otimes2}$ and the right product $\ket{0}^{\otimes2}$, $\ket{00}$ or $\ket{0}\ket{0}$. One may therefore express \eqref{eq:quantparallel} as a series of transformations applied to the 4-vector on the very right. Letting $\ket{10}$, $\ket{01}$ and $\ket{11}$ take on analogous meanings to $\ket{00}$, an immediate result of \eqref{eq:twoqubits} is that
\begin{align}\label{eq:twodmult}
	\H^{\otimes2} \ket{00} = \frac{1}{2}\Big(\ket{00} + \ket{01} + \ket{10} + \ket{11}\Big) \, .
\end{align}
The action of $\H^{\otimes2}$ transforms $\ket{00}$ into a superposition of the states $\ket{00}$, $\ket{10}$, $\ket{01}$ and $\ket{11}$, where the probability of selecting each is a uniform $(1/2)^2=1/4$. Writing so many $0$s and $1$s is tedious, so an alternative notation becomes preferable. Exchanging $\ket{0}$ for $\ket{00}$, $\ket{1}$ for $\ket{01}$, $\ket{2}$ for $\ket{10}$, and $\ket{3}$ for $\ket{11}$, \eqref{eq:twodmult} becomes the more succinct
\begin{align*}
	\H^{\otimes2} \ket{00} = \frac{1}{2}\Big(\ket{0} + \ket{1} + \ket{2} + \ket{3}\Big) \, .
\end{align*}
This formula extends generally to operations over $n$ qubits.  Now letting $N=2^n$, we have
\begin{align*}%\label{eq:superpos}
	\H^{\otimes n} \ket{0}^{\otimes n} = \frac{1}{\sqrt{N}}\Big(\ket{0} + \ket{1} + \dots + \ket{N-1}\Big) =  \frac{1}{\sqrt{N}} \sum_{x=0}^{N-1} \ket{x}=:\ket{h} \, ,
\end{align*}
and we call $\ket{h}$ a \emph{uniform superposition} over the states $\ket{0},\dots,\ket{N-1}$. The many-qubit analogue for the quantum parallelism of \eqref{eq:quantparallel} is then
\begin{align*}%\label{eq:quantparallel2}
	\ket{0}^{\otimes n}\ket{0} \longrightarrow \left(\frac{1}{\sqrt{N}} \sum_{x=0}^{N-1} \ket{x} \right) \ket{0}  \longrightarrow \frac{1}{\sqrt{N}} \sum_{x=0}^{N-1} \ket{x} \ket{f(x)} \, .
\end{align*}

\section{Gumbel-max}\label{sec:gm}

We wish to randomly select a single element from the set $\{0,1,\dots,P\}$ with probability proportional to the unnormalized probabilities $\ppi^*=(\pi_0^*,\pi_1^*,\dots,\pi_P^*)$.  That is, there exists a $c>0$ such that $\ppi^* = c \ppi$, for $\ppi$ a valid probability vector, but we only have access to $\ppi^*$.  Define $\llambda:= \log \ppi$ and $\llambda^*:= \log \ppi^*=\log\ppi + \log c$,  and assume that $z_0,\dots,z_P\stackrel{iid}{\sim}Gumbel(0,1)$. Then, the probability density function $g(\cdot)$ and cumulative distribution function $G(\cdot)$ for each individual $z_p$ are
\begin{align*}
	g(z_p) = \exp\big(-z_p- \exp(-z_p)\big)  \quad \mbox{and} \quad G(z_p) = \exp \big( -\exp(-z_p) \big)\, .
\end{align*}
Now, defining the random variables $\alpha^*_p:=\lambda^*_p+z_p$, $\alpha_p:=\lambda_p+z_p$ and  
\begin{align*}
	\hat{p} = \argmax_{p=0,\dots,P}\: \alpha^*_p \, ,
\end{align*}
we have the result
\begin{align*}
	\mbox{Pr} (\hat{p} = p) = \pi_p \, .
\end{align*}
In words, the procedure of adding Gumbel noise to unnormalized log-probabilities and taking the index of the maximum produces a random variable that follows the discrete distribution over $\ppi$.  Moving from left to right:
\begin{align*}
	\mbox{Pr} (\hat{p} = p) &= \mbox{Pr} (\alpha^*_p > \alpha^*_{p'}, \, \forall p' \neq p ) \\
	&= \mbox{Pr} (\alpha_p +\log c > \alpha_{p'} + \log c, \, \forall p' \neq p ) 
	= \mbox{Pr} (\alpha_p  > \alpha_{p'} , \, \forall p' \neq p ) \\
	&= \int_{-\infty}^\infty \prod_{p'\neq p} \mbox{Pr} (\alpha_p > \alpha_{p'}| \alpha_p) g(\alpha_p-\lambda_p) \, \dd \alpha_p 
	= \int_{-\infty}^\infty \prod_{p'\neq p} G(\alpha_p-\lambda_{p'}) g(\alpha_p-\lambda_p) \, \dd \alpha_p \\
	&=  \int_{-\infty}^\infty  \prod_{p'\neq p} \exp\big( -\exp(\lambda_{p'}-\alpha_p)\big) \exp\big(-\alpha_p+\lambda_p - \exp(-\alpha_p+\lambda_p) \big) \, \dd \alpha_p \, .
\end{align*}
Recalling that $\lambda_{p'}=\log \pi_{p'}$, we exponentiate the logarithms, and the integral becomes
\begin{align*}
	& \pi_p \int_{-\infty}^\infty  \prod_{p'\neq p} \exp\big( -\pi_{p'}\exp(-\alpha_p)\big) \exp  \big(-\alpha_p - \pi_p\exp(-\alpha_p) \big) \, \dd \alpha_p \\
	&= \pi_p \int_{-\infty}^\infty \exp  (-\alpha_p ) \exp\big( -\sum_{p'=0}^P\pi_{p'}\exp(-\alpha_p)\big)  \, \dd \alpha_p  \\
	&= \pi_p \int_{-\infty}^\infty \exp  (-\alpha_p ) \exp\big( -\exp(-\alpha_p)\big)  \, \dd \alpha_p  = \pi_p \, ,
\end{align*}
where the final equality follows easily from a change of variables. 

\section{Mixing of parallel MCMC and QPMCMC}\label{sec:mixing}

 \begin{figure}[!t]
	\centering
	\includegraphics[width=0.7\linewidth]{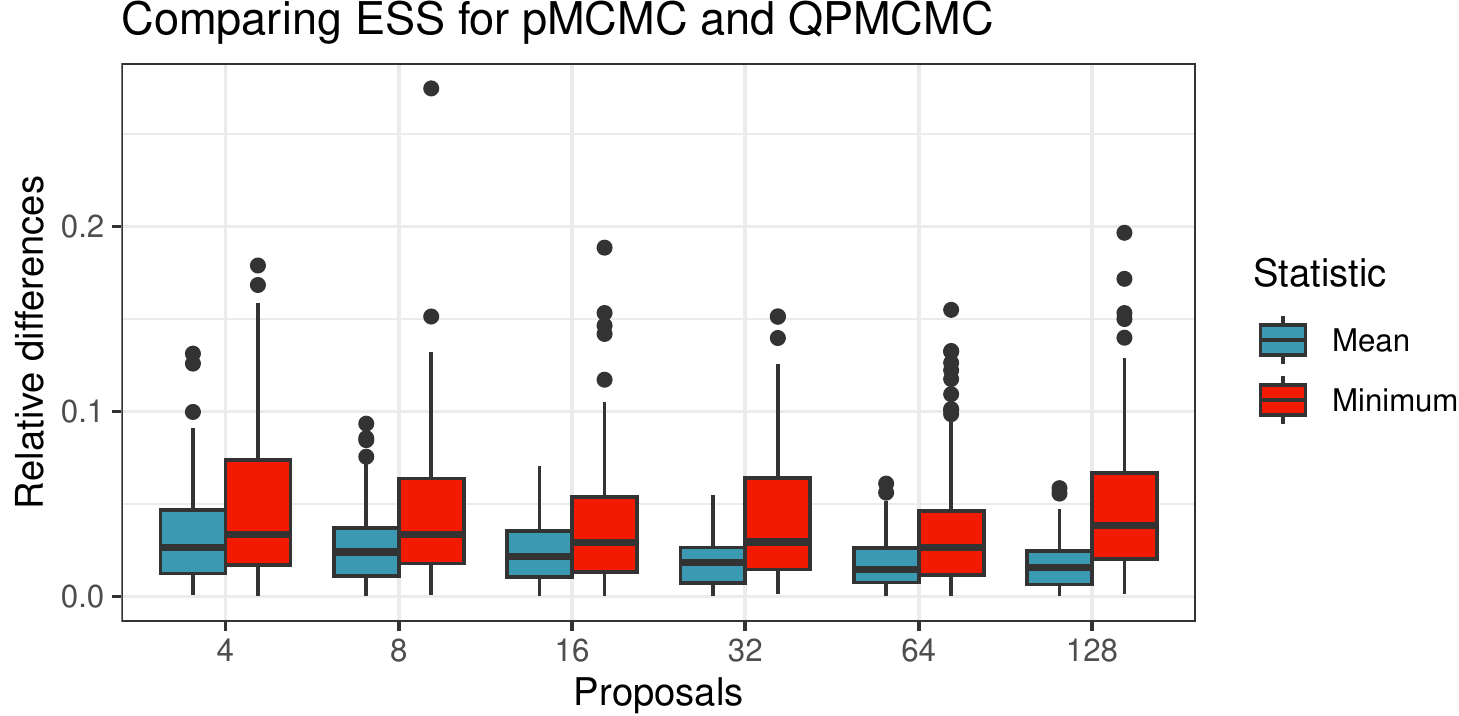}
	\caption{Relative differences between effective sample sizes (ESS) for parallel MCMC (pMCMC) and quantum parallel MCMC (QPMCMC) across a range of proposal counts. We target a 10-dimensional standard normal distribution.  For each algorithm and each proposal setting, we generate 100 independent chains of length 10,000 and calculate the mean and minimum ESS across dimensions.}\label{fig:mixing} 
\end{figure}

To ascertain whether QPMCMC mixes differently compared to conventional parallel MCMC, we run both algorithms for a range of proposal counts to sample from a 10-dimensional standard normal distribution. For each algorithm and proposal setting, we run 100 independent chains for 10,000 iterations and obtain effective sample sizes ESS$_d$ for $d \in \{1,\dots,10\}$.   We then calculate the relative differences between the means and minima of one chain generated using parallel MCMC and one chain generated using QPMCMC; for example:
\begin{align*}
	\mbox{Relative difference between means} \: \overline{\mbox{ESS}}_{(\cdot)} = \frac{\left|\overline{\mbox{ESS}}_{pMCMC} - \overline{\mbox{ESS}}_{QPMCMC}\right| }{\overline{\mbox{ESS}}_{pMCMC} }
\end{align*}
Figure \ref{fig:mixing} shows results.  In general, the majority relative differences are small.  For both statistics, mean relative differences are less than 0.05, regardless of proposal count.  Again for both statistics, more than 75\% of the independent runs result in relative differences below 0.1.  We note that relative differences between means (blue) appear to decrease with the number of proposals, but the same does not hold for relative differences between minima (red).

\section{Note on simulations}\label{sec:sim}

We use \textsc{R} \citep{rlang}, \textsc{Python} \citep{vanrossum1995python}, \textsc{TensorFlow} \citep{abadi2016tensorflow} and  \textsc{TensorFlow Probability MCMC} \citep{lao2020tfp} in our simulations and the \textsc{ggplot2} package to generate all figures \citep{ggplot}.   In \textsc{R}, we use the package \textsc{coda} \citep{coda} to calculate effective sample sizes.  In \textsc{Python}, we use a built-in function from \textsc{TensorFlow Probability MCMC}.  The primary color palette comes from \citet{wes}.

All simulations rely on the fact that the Grover iterations of \eqref{eq:grov} manipulate the probability amplitudes in a deterministic manner. For example, the following \textsc{R} code specifies $N$ uniform probability amplitudes (\verb|sqrtProbs|), performs $I$ Grover iterations and measures a single state according to the resulting probabilities.
%\begin{minted}{R}
%	sqrtProbs <- rep(1/sqrt(N),N); i <- 1
%	while (i <= I) {
%		sqrtProbs  <- (1-2*marks)*sqrtProbs
%		sqrtProbs  <- -sqrtProbs + 2*mean(sqrtProbs)
%		i <- i + 1
%	}
%	measurement <- sample(size=1, x=1:N, prob=sqrtProbs^2)
%\end{minted}
\begin{Verbatim}[commandchars=\\\{\}]
	\PYG{n}{sqrtProbs} \PYG{o}{\PYGZlt{}\PYGZhy{}} \PYG{n+nf}{rep}\PYG{p}{(}\PYG{l+m}{1}\PYG{o}{/}\PYG{n+nf}{sqrt}\PYG{p}{(}\PYG{n}{N}\PYG{p}{),}\PYG{n}{N}\PYG{p}{);} \PYG{n}{i} \PYG{o}{\PYGZlt{}\PYGZhy{}} \PYG{l+m}{1}
	\PYG{n+nf}{while }\PYG{p}{(}\PYG{n}{i} \PYG{o}{\PYGZlt{}=} \PYG{n}{I}\PYG{p}{)} \PYG{p}{\PYGZob{}}
	\PYG{n}{sqrtProbs}  \PYG{o}{\PYGZlt{}\PYGZhy{}} \PYG{p}{(}\PYG{l+m}{1\PYGZhy{}2}\PYG{o}{*}\PYG{n}{marks}\PYG{p}{)}\PYG{o}{*}\PYG{n}{sqrtProbs}
	\PYG{n}{sqrtProbs}  \PYG{o}{\PYGZlt{}\PYGZhy{}} \PYG{o}{\PYGZhy{}}\PYG{n}{sqrtProbs} \PYG{o}{+} \PYG{l+m}{2}\PYG{o}{*}\PYG{n+nf}{mean}\PYG{p}{(}\PYG{n}{sqrtProbs}\PYG{p}{)}
	\PYG{n}{i} \PYG{o}{\PYGZlt{}\PYGZhy{}} \PYG{n}{i} \PYG{o}{+} \PYG{l+m}{1}
	\PYG{p}{\PYGZcb{}}
	\PYG{n}{measurement} \PYG{o}{\PYGZlt{}\PYGZhy{}} \PYG{n+nf}{sample}\PYG{p}{(}\PYG{n}{size}\PYG{o}{=}\PYG{l+m}{1}\PYG{p}{,} \PYG{n}{x}\PYG{o}{=}\PYG{l+m}{1}\PYG{o}{:}\PYG{n}{N}\PYG{p}{,} \PYG{n}{prob}\PYG{o}{=}\PYG{n}{sqrtProbs}\PYG{o}{\PYGZca{}}\PYG{l+m}{2}\PYG{p}{)}
\end{Verbatim}
Of course, simulating these iterations on a classical computer requires precomputing the values of \verb|marks| beforehand.   All code is available online at \url{https://github.com/andrewjholbrook/qpMCMC}.

\bibliographystyle{sysbio}
\bibliography{refs}

%%%%%%%%%%%%%%%%%%%%%%%%%%%%%%%%%%%%%%%%%%%%%%%%%%%%%%%%%%%%

\end{document}